\documentclass[reqno,10pt,centertags]{amsart}
\usepackage{amsmath,amsthm,amscd,amssymb,latexsym,esint,upref,stmaryrd,
enumerate,color,verbatim,yfonts}
\usepackage{color}
\usepackage{hyperref}
\newcommand*{\mailto}[1]{\href{mailto:#1}{\nolinkurl{#1}}}
\newcommand{\arxiv}[1]{\href{http://arxiv.org/abs/#1}{arXiv:#1}}

\newcommand{\C}{{\mathbb C}}

\newcommand{\bbC}{{\mathbb{C}}}

\newcommand{\bbN}{{\mathbb{N}}}

\newcommand{\bbR}{{\mathbb{R}}}

\newcommand{\bbZ}{{\mathbb{Z}}}

\newcommand{\bsA}{{\boldsymbol{A}}}
\newcommand{\bsB}{{\boldsymbol{B}}}
\newcommand{\bsC}{{\boldsymbol{C}}}
\newcommand{\bsD}{{\boldsymbol{D}}}

\newcommand{\bsH}{{\boldsymbol{H}}}
\newcommand{\bsI}{{\boldsymbol{I}}}

\newcommand{\bsT}{{\boldsymbol{T}}}

\newcommand{\cB}{{\mathcal B}}
\newcommand{\cC}{{\mathcal C}}

\newcommand{\cF}{{\mathcal F}}

\newcommand{\cH}{{\mathcal H}}

\newcommand{\cJ}{{\mathcal J}}
\newcommand{\cK}{{\mathcal K}}

\newcommand{\beq}{\begin{equation}}
\newcommand{\enq}{\end{equation}}

\DeclareMathOperator{\esssup}{ess.sup}

\DeclareMathOperator{\dom}{dom}

\DeclareMathOperator{\tr}{tr}

\DeclareMathOperator*{\nlim}{n-lim}
\DeclareMathOperator*{\slim}{s-lim}
\DeclareMathOperator*{\wlim}{w-lim}

\DeclareMathOperator*{\sgn}{sgn}

\renewcommand{\Re}{\text{\rm Re}}
\renewcommand{\Im}{\text{\rm Im}}
\renewcommand{\ln}{\text{\rm ln}}

\newcommand{\loc}{\operatorname{loc}}
\newcommand{\locunif}{\operatorname{loc \, unif}}
\newcommand{\Lxi}{\xi_L}

\newcommand{\ind}{\operatorname{index}}
\newcommand{\no}{\notag}
\newcommand{\lb}{\label}
\newcommand{\f}{\frac}

\newcommand{\ol}{\overline}

\newcommand{\wti}{\widetilde}
\newcommand{\Oh}{O}

\newcommand{\hatt}{\widehat} 
\newcommand{\bi}{\bibitem}

\let\geq\geqslant
\let\leq\leqslant

\makeatletter
\def\theequation{\@arabic\c@equation}

\allowdisplaybreaks 
\numberwithin{equation}{section}

\newtheorem{theorem}{Theorem}[section]

\newtheorem{lemma}[theorem]{Lemma}
\newtheorem{corollary}[theorem]{Corollary}
\newtheorem{definition}[theorem]{Definition}
\newtheorem{hypothesis}[theorem]{Hypothesis}

\theoremstyle{remark}
\newtheorem{remark}[theorem]{Remark}

\begin{document}

\title[On Index Theory for Non-Fredholm Operators]{On Index Theory for Non-Fredholm Operators: A $(1+1)$-Dimensional Example}

\author[A.\ Carey]{Alan Carey}  
\address{Mathematical Sciences Institute, Australian National University, 
Kingsley St., Canberra, ACT 0200, Australia 
and School of Mathematics and Applied Statistics, University of Wollongong, NSW, Australia,  2522}  
\email{\mailto{acarey@maths.anu.edu.au}}
\urladdr{\url{http://maths.anu.edu.au/~acarey/}}
  
\author[F.\ Gesztesy]{Fritz Gesztesy}  
\address{Department of Mathematics,
University of Missouri, Columbia, MO 65211, USA}
\email{\mailto{gesztesyf@missouri.edu}}
\urladdr{\url{https://www.math.missouri.edu/people/gesztesy}}

\author[G.\ Levitina]{Galina Levitina} 
\address{School of Mathematics and Statistics, UNSW, Kensington, NSW 2052,
Australia} 
\email{\mailto{g.levitina@student.unsw.edu.au}}

\author[D.\ Potapov]{Denis Potapov}
\address{School of Mathematics and Statistics, UNSW, Kensington, NSW 2052,
Australia} 
\email{\mailto{d.potapov@unsw.edu.au}}

\author[F.\ Sukochev]{Fedor Sukochev}
\address{School of Mathematics and Statistics, UNSW, Kensington, NSW 2052,
Australia} 
\email{\mailto{f.sukochev@unsw.edu.au}}

\author[D.\ Zanin]{Dima Zanin}
\address{School of Mathematics and Statistics, UNSW, Kensington, NSW 2052,
Australia} 
\email{\mailto{d.zanin@unsw.edu.au}}


\thanks{To appear in {\it Math. Nachrichten}{\bf }.}

\date{\today}

{\scriptsize{
\begin{abstract} 
Using the general formalism of  \cite{GLMST11}, a study of index theory for non-Fredholm operators was initiated in \cite{CGPST14a}. Natural examples arise from
$(1+1)$-dimensional differential operators using the model operator $\bsD_{\bsA}$ in $L^2(\bbR^2; dt dx)$  of the type $\bsD_\bsA^{} = \f{d}{dt} + \bsA$, where $\bsA = \int^{\oplus}_{\bbR} dt \, A(t)$, and 
the family of self-adjoint operators $A(t)$ in $L^2(\bbR; dx)$ studied here
is explicitly given by 
\begin{equation*}
A(t) = - i \f{d}{dx} + \theta(t) \phi(\cdot), \quad t \in \bbR.
\end{equation*} 
Here $\phi: \bbR \to \bbR$ has to be integrable on $\bbR$ and 
$\theta: \bbR \to \bbR$ tends to zero as $t \to - \infty$ and to $1$ as $t \to + \infty$ 
(both functions are subject to additional hypotheses). In particular, $A(t)$, 
$t \in \bbR$, has asymptotes (in the norm resolvent sense)   
\begin{equation*} 
A_- = - i \f{d}{dx}, \quad A_+ = - i \f{d}{dx} + \phi(\cdot)  
\end{equation*} 
as $t \to \mp \infty$, respectively. 

The interesting feature is that $\bsD_\bsA^{}$  violates the relative trace class condition 
introduced in \cite[Hypothesis~2.1\,$(iv)$]{CGPST14a}. 
A new approach adapted to differential operators
of this kind is given here using an approximation technique. The approximants do 
fit the framework of \cite{CGPST14a} enabling the following results to be obtained.
Introducing 
$\bsH_1 = \bsD_\bsA^* \bsD_\bsA^{}$,  $\bsH_2 = \bsD_\bsA^{} \bsD_\bsA^*$, 
we recall that the resolvent regularized Witten index of $\bsD_\bsA^{}$, denoted by 
$W_r(\bsD_{\bsA}^{})$, is defined by  
\begin{equation*}
W_r(\bsD_\bsA^{}) = \lim_{\lambda \uparrow 0} (- \lambda) 
\tr_{L^2(\bbR^2; dtdx)}\big((\bsH_1 - \lambda \bsI)^{-1}
- (\bsH_2 - \lambda \bsI)^{-1}\big),         
\end{equation*} 
whenever this limit exists. In the concrete example at hand, we prove 
\begin{equation*}
W_r(\bsD_\bsA^{}) = \xi(0_+; \bsH_2, \bsH_1) = 
\xi(0; A_+, A_-) = \f{1}{2 \pi} \int_{\bbR} dx \, \phi(x).      
\end{equation*}
Here $\xi(\, \cdot \, ; S_2, S_1)$ denotes the spectral shift operator for the pair of 
self-adjoint operators $(S_2,S_1)$, and we employ the normalization, 
$\xi(\lambda; \bsH_2, \bsH_1) = 0$, $\lambda < 0$.
\end{abstract}
}}

\maketitle


{\scriptsize{\tableofcontents}}

\section{Introduction}  \lb{s1}

This paper is motivated by \cite{CGPST14a} where results on an index theory for certain non-Fredholm operators are described using the model operator formalism in \cite{GLMST11}. 
The latter paper was particularly motivated by \cite{Pu08} which, in turn, was motivated by \cite{RS95}, which relates the Fredholm index and spectral flow for operators with compact resolvent. The model operators considered there provide prototypes  for more complex situations. They arise in connection with investigations of the Maslov index, Morse theory, Floer homology, Sturm oscillation theory, etc. The 
principle aim in \cite{Pu08} and \cite{GLMST11} was to extend the compact resolvent results in \cite{RS95} to a relatively trace class perturbation approach, permitting essential spectra. 

To introduce the situation of \cite{CGPST14a}, let $\{A(t)\}_{t\in\bbR}$ be a family of self-adjoint operators in 
the complex, separable Hilbert space $\cH$, subject to the relative trace class condition.
In particular, it is assumed that self-adjoint limiting operators
\begin{equation} 
  A_+=\lim_{t\to+\infty}A(t), \quad A_-=\lim_{t\to-\infty}A(t)    \lb{1.2}
\end{equation} 
exist in $\cH$ in the norm resolvent sense
 and that $(A_+-A_-)\big(A_-^2 + I_{\cH}\big)^{-1/2}$ is trace class (for precise conditions see
 Hypothesis 2.1 of \cite{GLMST11}).
We denote by $\bsA$ the operator in 
$L^2(\bbR;\cH)$ defined by   
\begin{align}
&(\bsA f)(t) = A(t) f(t) \, \text{ for a.e.\ $t\in\bbR$,}   \no \\
& f \in \dom(\bsA) = \bigg\{g \in L^2(\bbR;\cH) \,\bigg|\,
g(t)\in \dom(A(t)) \text{ for a.e.\ } t\in\bbR;     \lb{1.1} \\
& \quad t \mapsto A(t)g(t) \text{ is (weakly) measurable;} \,  
\int_{\bbR} dt \, \|A(t) g(t)\|_{\cH}^2 < \infty\bigg\}.   \no 
\end{align}
Of course, $\bsA = \int^{\oplus}_{\bbR} dt \, A(t)$. 

Next, we introduce the model operator 
\begin{equation}
\bsD_\bsA^{} = \f{d}{dt} + \bsA,
\quad \dom(\bsD_\bsA^{})= \dom(d/dt) \cap \dom(\bsA_-),     \lb{1.5}   \\
\end{equation}
and the associated nonnegative, self-adjoint operators 
\begin{equation}
\bsH_1 = \bsD_\bsA^* \bsD_\bsA^{}, \quad \bsH_2 = \bsD_\bsA^{} \bsD_\bsA^*,   
\lb{1.6}
\end{equation}
in $L^2(\bbR;\cH)$. (Here $\bsA_-$  in $L^2(\bbR;\cH)$ represents the self-adjoint constant fiber operator defined according to \eqref{1.1}, with $A(t)$ replaced by the asymptote $A_-$.)

Assuming that $A_-$ and $A_+$ are boundedly invertible, we also recall (cf.\ \cite{GLMST11}) that
$\bsD_\bsA^{}$ is a Fredholm operator in $L^2(\bbR;\cH)$. Moreover, as shown in 
\cite{GLMST11} (and earlier in \cite{Pu08} under a simpler set of hypotheses on the family 
$A(\cdot)$), the Fredholm index of $\bsD_\bsA^{}$ may then be computed as follows,   
\begin{align}
\ind (\bsD_\bsA^{}) & = \xi(0_+; \bsH_2, \bsH_1)   
 = \xi(0; A_+, A_-).   \lb{1.10} 
\end{align}
Here $\xi(\, \cdot \, ; S_2,S_1)$ denotes the spectral shift function for 
the pair of self-adjoint operators $(S_2, S_1)$. Whenever $S_j$, $j=1,2$, are bounded from below, 
we adhere to the normalization
\begin{equation}
\xi(\lambda ; S_2,S_1) = 0 \, \text{ for } \, \lambda < \inf(\sigma(S_1) \cup \sigma(S_2)),  
\end{equation}  
in particular, $\xi(\lambda; \bsH_2, \bsH_1) = 0$, $\lambda < 0$.

The new direction developed in \cite{CGPST14a} focuses on the model operator 
$\bsD_\bsA^{}$ in $L^2(\bbR;\cH)$ whenever the latter ceases to be Fredholm. 
First, we recall the definition of the Witten index as studied 
in \cite{BGGSS87} and \cite{GS88}: 
\begin{equation}
W_r(\bsD_\bsA^{}) = \lim_{\lambda \uparrow 0} (- \lambda) 
\tr_{L^2(\bbR;\cH)}\big((\bsH_1 - \lambda \bsI_{L^2(\bbR;\cH)})^{-1} - (\bsH_2 - \lambda \bsI_{L^2(\bbR;\cH)})^{-1}\big),     \lb{1.13a}
\end{equation} 
whenever the limit exists. Here, the subscript ``r'' refers to the resolvent regularization used; other regularizations, for instance, semigroup (heat kernel) based ones, are 
possible (cf., \cite{CGPST14a}).  

If $\bsD_\bsA^{}$ is Fredholm (and of course the necessary trace class conditions in \eqref{1.13a} are satisfied), one has consistency with the Fredholm index, 
$\ind(\bsD_\bsA^{})$ of $\bsD_\bsA^{}$. In addition, under appropriate spectral assumptions  the following connection between Fredholm, respectively, Witten indices and the underlying spectral shift function applies
\begin{equation}
\mbox{Index }(\bsD_\bsA^{}) =
W_r(\bsD_\bsA^{}) = \xi(0_+; \bsH_2, \bsH_1).    \lb{1.13c}
\end{equation} 
Most importantly, $W_r(\bsD_\bsA^{})$ exhibits invariance properties under 
additive, relatively trace class perturbations (apart from some additional technical 
hypotheses). This is sometimes dubbed topological invariance of the Witten index 
in the pertinent literature (see, e.g., \cite{BGGSS87}, \cite{Ca78}, \cite{GLMST11}, \cite{GS88}, and the references therein).

Originally, index regularizations such as \eqref{1.13a} were studied in the context of 
supersymmetric quantum mechanics in the physics literature in the 1970's and 1980's, see, \cite{CGPST14a} for details.
A theory for non-Fredholm operators was initiated in \cite{BGGSS87} and \cite{GS88}, 
however, it was technically quite
a formidable problem at that time to construct a wide range of examples. This paper produces
examples for which the Witten index may be calculated explicitly in the non-Fredholm case.

The results of \cite{CGPST14a} for the
specific model operator $\bsD_\bsA^{}$ in $L^2(\bbR;\cH)$ are as follows. 
Assume that $0$ is a right 
and a left Lebesgue point of $\xi(\,\cdot\,\, ; A_+, A_-)$, we denote this by $\Lxi(0_+; A_+,A_-)$ 
and $\Lxi(0_-; A_+, A_-)$, respectively: then it is 
also a right Lebesgue point of $\xi(\,\cdot\,\, ; \bsH_2, \bsH_1)$, which we denote by 
$\Lxi(0_+; \bsH_2, \bsH_1)$. Under this right/left Lebesgue point 
assumption on $\xi(\,\cdot\,\, ; A_+, A_-)$ (and when $\bsD_\bsA^{}$ ceases to be Fredholm), 
the  principal new result of \cite{CGPST14a} then reads, 
\begin{align} 
W_r(\bsD_\bsA^{}) & = \Lxi(0_+; \bsH_2, \bsH_1)      \lb{1.15} \\ 
& = [\Lxi(0_+; A_+,A_-) + \Lxi(0_-; A_+, A_-)]/2.    \lb{1.16}
\end{align}

Now we come to the central observation. In dimension one we may choose the Clifford generator
so that the Dirac operator, densely defined on $L^2(\bbR) \otimes \bbC^2$, takes the form
\begin{equation} 
\begin{pmatrix} i\frac{d}{d x}  & 0 \\
0 & -i \frac{d}{d x} \end{pmatrix}.
\end{equation} 
In the simplest case a connection is represented by a function $\phi$ on $\mathbb R$  so that 
\begin{equation} 
\begin{pmatrix} i \frac{d}{d x} -\phi & 0 \\
0 & -i \frac{d}{d x} +\phi
\end{pmatrix}
\end{equation} 
is of the form of a ``coupled Dirac operator''. (More generally, $\phi$ can be matrix-valued.) Thus, in dimension one we may study each component separately and hence in this paper we consider model operators of the form 
$-i \frac{d}{d x} +\phi$.

Next, let $A_-$ be the flat space Dirac operator on spinor-valued functions on $\bbR^d$, $d \in \bbN$, and perturb it by an operator of multiplication by a function $\phi:\bbR^d \to \bbR$ (matrix-valued functions also provide examples) to define $A_+ = A_- + \phi$. It follows from the discussion in Remark~(c) of \cite[Ch.~4]{Si05} that, for $\phi$ of sufficiently rapid decay at $\pm\infty$, 
$(A_+-A_-)(A_-^2 + I_{\cH})^{-s/2}$ is trace class for $s>d$, but for no lesser 
value of $s$. Thus, for $d=1$, this means that with
$A_- = -i \frac{d}{d x}$, $A_+ = A_- + \phi$, one infers that $\phi{(I_{\cH}+A_-^2)^{-s/2}}$
is trace class only for $s>1$. Thus, for all dimensions $d \in \bbN$, the relative trace class perturbation assumption introduced in \cite[Hypothesis~2.1\,$(iv)$]{CGPST14a} is violated.

We show here that by replacing $\phi$ by certain pseudodifferential approximating operators, this  relatively trace class perturbation condition is restored and the results of our earlier papers on spectral shift functions (\cite{CGPST14}, \cite{CGPST14a}) are available for the approximating
operators.   The problem we then face is to take limits of our spectral shift functions
as the approximating operators converge (in an appropriate strong, resp., norm resolvent sense) 
to the original operators. We find ourselves following in the footsteps of Fredholm theory history here for we are able to control this limit in one dimension only so as to obtain an index theorem for non-Fredholm operators that is related to the classical Gohberg--Krein theory \cite{GK60}. Even in this 
special case the analysis is both subtle and involved.

The results in this paper have motivated an abstract approach to this circle of ideas in 
\cite{CGLS14} and these indicate that additional tools have to be employed in higher dimensions (cf. also \cite{CGK15}). 
 We also remark that trace formulas related to a matrix-valued extension of 
the model discussed in this paper will appear in a companion paper \cite{CGLPSZ14}.
Our main result is stated in part IV of the next section.

\section{The Strategy Employed and Statement of Results} \lb{s2} 

In this section we briefly outline the principal new strategy employed in this paper that 
permits us to circumvent the relative trace class hypotheses used in 
\cite[Hypothesis~2.1]{GLMST11} and \cite{CGPST14a}.

$\mathbf{I}$. We assume Hypothesis \ref{h3.1} and introduce in $L^2(\bbR)$, the family of 
self-adjoint operators
\begin{equation}
A(t) = - i \f{d}{dx} + \theta(t) \phi(\cdot), \quad 
\dom(A(t)) = W^{1,2}(\bbR), \quad t \in \bbR,
\end{equation} 
and 
\begin{equation}
A_- = - i \f{d}{dx}, \quad A_+ = - i \f{d}{dx} + \phi, \quad 
\dom(A_{\pm}) = W^{1,2}(\bbR),    \lb{} 
\end{equation}
as well as families of self-adjoint bounded operators 
\begin{equation}
B(t) = \theta(t) \phi, \quad B'(t) = \theta'(t) \phi, \quad t \in \bbR,
\end{equation}
such that 
\begin{equation}
A(t) = A_- + B(t), \quad t \in \bbR. 
\end{equation} 
Introduce in $L^2(\bbR^2)$,
\begin{equation}
\bsD_\bsA^{} = \f{d}{dt} + \bsA,
\quad \dom(\bsD_\bsA^{})= W^{1,2}(\bbR^2).   \lb{DA}
\end{equation}
For simplicity, we tacitly identify $L^2(\bbR^2; dt dx)$ and 
$L^2\big(\bbR; dt; L^2(\bbR;dx)\big)$, and simply abbreviate it by $L^2(\bbR^2)$. 
Then $\bsD_\bsA^{}$ is densely defined and closed in 
$L^2(\bbR^2)$ and the adjoint operator $\bsD_\bsA^*$ of 
$\bsD_\bsA^{}$ is then given by
\begin{equation}
\bsD_\bsA^*=- \f{d}{dt} + \bsA, \quad
\dom(\bsD_\bsA^*) = W^{1,2}(\bbR^2).   
\end{equation}
Next, introduce in $L^2(\bbR^2)$ the operators $\bsH_j$, $j=1,2$, by
\begin{align}
& \bsH_1 = \bsD_{\bsA}^{*} \bsD_{\bsA}^{} = - \f{\partial^2}{\partial t^2} - \f{\partial^2}{\partial x^2} 
- 2 i \theta(t) \phi(x) \f{\partial}{\partial x}   \no \\ 
& \hspace*{2.6cm} - \theta'(t) \phi(x) - i \theta(t) \phi'(x) + \theta^2(t) \phi(x)^2,   \lb{2.7a}  \\
& \bsH_2 = \bsD_{\bsA}^{} \bsD_{\bsA}^{*} = - \f{\partial^2}{\partial t^2} - \f{\partial^2}{\partial x^2} 
- 2 i \theta(t) \phi(x) \f{\partial}{\partial x}   \no \\ 
& \hspace*{2.6cm} + \theta'(t) \phi(x) - i \theta(t) \phi'(x) + \theta^2(t) \phi(x)^2,    \lb{2.8a} \\
& \dom(\bsH_1) =  \dom(\bsH_2) = W^{2,2}\big(\bbR^2\big).   \lb{2.9a}
\end{align} 
In particular,
introducing $\bsB$ and $\bsA^{\prime} = \bsB^{\prime}$ in terms of the bounded operator families 
$B(t)$, $B'(t)$, $t \in \bbR$, in analogy to \eqref{1.S}, one can decompose 
$\bsH_j$, $j=1,2$, as follows:
\begin{align}
\bsH_j &= \f{d^2}{dt^2} + \bsA^2 + (-1)^j \bsA^{\prime}    \\
&= \bsH_0 + \bsB \bsA_- + \bsA_- \bsB + \bsB^2 + (-1)^j \bsB^{\prime}, \quad j =1,2, 
\end{align}
with 
\begin{equation}
\bsA = \bsA_- + \bsB, \quad \dom(\bsA) = \dom(\bsA_-), 
\end{equation}
and $\bsA_-$  in $L^2(\bbR^2)$ representing 
the self-adjoint (constant fiber) operator defined by 
\begin{align}
&(\bsA_- f)(t) = A_- f(t) \, \text{ for a.e.\ $t\in\bbR$,}   \no \\
& f \in \dom(\bsA_-) = \bigg\{g \in L^2(\bbR^2) \,\bigg|\,
g(t, \cdot)\in W^{1,2}(\bbR) \text{ for a.e.\ } t\in\bbR,    \no \\
& \quad t \mapsto A_- g(t,\cdot) \text{ is (weakly) measurable,} \,  
\int_{\bbR} dt \, \|A_- g(t)\|_{L^2(\bbR;dx)}^2 < \infty\bigg\}.    \lb{}
\end{align} 

According to a tradition in mathematical physics, we called the model represented 
by $\bsD_\bsA^{}$ a $(1+1)$-dimensional model due to the fact that the underlying 
variables $x \in \bbR$ and $t \in \bbR$ are both one-dimensional. (In future investigations 
we intend to study $(n+1)$-dimensional models in which $A_-$ represents an $n$-dimensional 
Dirac-type operator.)  

$\mathbf{II}$. Use the approximation
\begin{equation}
A_n(t) = A_- + \theta(t) \chi_n(A_-) \phi \chi_n(A_-), \quad n \in \bbN, \; t \in \bbR, 
\end{equation}
with  
\begin{equation}
\chi_n(\nu) = \f{n}{(\nu^2 + n^2)^{1/2}}, \quad \nu \in \bbR, \; n \in \bbN, 
\end{equation} 
such that 
\begin{equation}
\slim_{n \to \infty} \chi_n(A_-) = I.     \lb{} 
\end{equation} 
Then one concludes that
\begin{align}
& A_{-,n} = A_-, \quad   
 A_{+,n} = A_- + \chi_n(A_-) \phi \chi_n(A_-), \quad n \in \bbN,   \\
& A_n'(t) = B_n'(t) = \theta'(t) \chi_n(A_-) \phi \chi_n(A_-) \in \cB_1\big(L^2(\bbR)\big), 
\quad n \in \bbN, \; t \in \bbR,  \\
& \int_{\bbR} dt \, \|A_n'(t)\|_{\cB_1(L^2(\bbR))} \leq \|\chi_n(A_-) \phi \chi_n(A_-)\|_{\cB_1(L^2(\bbR))} \, 
\|\theta'\|_{L^1(\bbR)} < \infty.     \lb{} 
\end{align}
Thus, one also obtains,
\begin{align}
& \bsH_{j,n} = \f{d^2}{dt^2} + \bsA_n^2 + (-1)^j \bsA_n^{\prime}    \no \\
& \hspace*{8mm} = \bsH_0 + \bsB_n \bsA_- + \bsA_- \bsB_n 
+ \bsB_n^2 + (-1)^j \bsB_n^{\prime}, \\
& \dom(\bsH_{j,n}) = \dom(\bsH_0) = W^{2,2}(\bbR^2),  
\quad n \in \bbN, \;  j =1,2,  \no 
\end{align}
with
\begin{equation}
\bsB_n = \chi_n(\bsA_-) \bsB \chi_n(\bsA_-), \quad \bsA_n^{\prime} =
\bsB_n^{\prime} = \chi_n(\bsA_-) \bsB^{\prime} \chi_n(\bsA_-), \quad n \in \bbN. 
\end{equation} 

$\mathbf{III}$. As a consequence of step $\mathbf{II}$, obtain the approximate 
trace formula,
\begin{equation}
\int_{[0,\infty)} \f{\xi(\lambda; \bsH_{2,n}, \bsH_{1,n}) d\lambda}{(\lambda - z)^2} 
= \f{1}{2} \int_{\bbR} \f{\xi(\nu; A_{+,n}, A_-) d\nu}{(\nu^2 -z)^{3/2}}, \quad 
n \in \bbN, \; z \in \bbC \backslash [0,\infty)      \lb{1}
\end{equation}
(employing \cite{CGPST14} or \cite{Pu08}), which implies 
\begin{equation}
\xi(\lambda; \bsH_{2,n}, \bsH_{1,n}) = \f{1}{\pi} \int_{- \lambda^{1/2}}^{\lambda^{1/2}} 
\f{\xi(\nu; A_{+,n}, A_-) d \nu}{(\lambda - \nu^2)^{1/2}} \, 
\text{ for a.e.~$\lambda > 0$, $n\in\bbN$,}    \lb{2} 
\end{equation}
via a Stieltjes inversion argument. Here $\xi(\, \cdot \, ; S_2, S_1)$ denotes the 
spectral shift operator for the pair of 
self-adjoint operators $(S_2,S_1)$, and we employed the normalization, 
$\xi(\lambda; \bsH_{2,n}, \bsH_{1,n}) = 0$, $\lambda < 0$, $n \in \bbN$.

$\mathbf{IV}$. {\bf The main results}. Now we take the limits $n \to \infty$ in \eqref{1}. We use the trace norm convergence result in Theorem \ref{t3.7} in combination with a variety of Fredholm determinant facts to control the limit $n \to \infty$ of the left- and right-hand side of \eqref{1} 
to arrive at
\begin{equation}
\int_{[0,\infty)} \f{\xi(\lambda; \bsH_2, \bsH_1) d\lambda}{(\lambda - z)^2} 
= \f{1}{2} \int_{\bbR} \f{\xi(\nu; A_+, A_-) d\nu}{(\nu^2 -z)^{3/2}}, \quad 
z \in \bbC \backslash [0,\infty).      \lb{3}
\end{equation}
Relation \eqref{3} combined with a Stieltjes inversion argument then  implies the main formula of the paper (a Pushnitski-type relation between spectral shift functions):
\begin{equation}
\xi(\lambda; \bsH_2, \bsH_1) = \f{1}{\pi} \int_{- \lambda^{1/2}}^{\lambda^{1/2}} 
\f{\xi(\nu; A_+, A_-) d \nu}{(\lambda - \nu^2)^{1/2}} \, 
\text{ for a.e.~$\lambda > 0$}    \lb{4} 
\end{equation}
(i.e., formally, the limit of \eqref{2} as $n \to \infty$). Again, we employed the normalization, $\xi(\lambda; \bsH_2, \bsH_1) = 0$, $\lambda < 0$.

As a consequence of formula \eqref{4}, the principal result obtained for this $(1+1)$-dimensional 
example reads
\begin{equation}
W_r(\bsD_\bsA^{}) = \xi(0_+; \bsH_2, \bsH_1) = 
\xi(0; A_+, A_-) = \f{1}{2 \pi} \int_{\bbR} dx \, \phi(x).      \lb{5}
\end{equation}
with the most novel point being the middle equality.
Equations \eqref{3}, \eqref{4}, and \eqref{5} represent the analog of the principal 
results in \cite{GLMST11} and \cite{CGPST14a} for the non-Fredholm model 
operator \eqref{DA}. 

$\mathbf{V}$. {\bf Outline of the exposition}. 
Section \ref{s3} provides a detailed description of $\xi(\, \cdot \,; A_+, A_-)$ and  its approximating sequence, $\xi(\, \cdot \,; A_{+,n}, A_-)$, $n \in \bbN$.
There we establish the final equality in \eqref{5} using scattering theory (which only gives the value of
$\xi(\, \cdot \,; A_+, A_-)$ up to an undetermined additive integer constant) with the precise value  
being fixed by our approximation technique.
Section \ref{s4} begins with a discussion of the example and develops the approximation 
argument for the left hand side of \eqref{3}, \eqref{4}. 
Section \ref{s5} completes the proof of  \eqref{5}
by obtaining the stronger fact that the spectral shift function for the pair $(\bsH_2, \bsH_1)$ is constant and equals that of the spectral shift function for the pair $(A_+, A_-)$, that is, for a.e.~$\lambda > 0$ and a.e.~$\nu \in \bbR$,
\begin{equation}
\xi(\lambda; \bsH_2, \bsH_1) = \xi(\nu; A_+, A_-) = \f{1}{2 \pi} \int_{\bbR} dx \, \phi(x).     \lb{6}
\end{equation}  
Given this fact, the calculation of the Witten index \eqref{5} in Section \ref{s6} is straightforward. 
Appendix \ref{sA} collects a number of results on trace formulas and modified determinants of 
fundamental importance in Sections \ref{s3} and \ref{s5}. 

\smallskip

$\mathbf{VI}$. {\bf Notation.} We briefly summarize some of the notation used in this paper: 
Let $\cH$ be a separable complex Hilbert space, $(\cdot,\cdot)_{\cH}$ the scalar product in $\cH$
(linear in the second argument), and $I_{\cH}$ the identity operator in $\cH$.

Next, if $T$ is a linear operator mapping (a subspace of) a Hilbert space into another, then 
$\dom(T)$ and $\ker(T)$ denote the domain and kernel (i.e., null space) of $T$. 
The closure of a closable operator $S$ is denoted by $\ol S$. 
The spectrum, essential spectrum, discrete spectrum, point spectrum, and resolvent set 
of a closed linear operator in a Hilbert space will be denoted by $\sigma(\cdot)$, 
$\sigma_{\rm ess}(\cdot)$, $\sigma_{\rm d}(\cdot)$, $\sigma_{\rm p}(\cdot)$, and 
$\rho(\cdot)$, respectively. 

The convergence of bounded operators in the strong operator topology (i.e., pointwise limits) will be denoted by $\slim$, similarly, norm limits of bounded operators are denoted by $\nlim$. 

The strongly right continuous family of spectral projections of a self-adjoint 
operator $S$ in $\cH$ will be denoted by $E_S(\lambda)$, $\lambda \in \bbR$. 
(In particular, $E_S(\lambda) = E_S((-\infty,\lambda])$, 
$E_S([\lambda, \infty)) = I_{\cH} - E_S((-\infty,\lambda))$, $\lambda \in \bbR$, 
$E_S((a,b)) = E_S(b_-) - E_S(a)$, $(a,b) \subseteq \bbR$, etc.) 

The Banach spaces of bounded and compact linear operators on a separable complex Hilbert 
space $\cH$ are denoted by $\cB(\cH)$ and $\cB_\infty(\cH)$, respectively; the corresponding 
$\ell^p$-based Schatten--von Neumann trace ideals (cf.\ \cite[Ch.~III]{GK69}, \cite[Ch.~1]{Si05}) 
will be denoted by $\cB_p (\cH)$, with corresponding norm denoted by  
$\|\cdot\|_{\cB_p(\cH)}$, $p \geq 1$ (and defined in terms of the $\ell^p$-norm of the singular values 
of the operator in question). Moreover, ${\det}_{\cH}(I_\cK-A)$, and $\tr_{\cH}(A)$ 
denote the standard Fredholm determinant and the corresponding trace of a trace class operator $A\in\cB_1(\cH)$. 
Similarly, ${\det}_{p,\cH}(I_\cK-B)$ represents the $p$th modified Fredholm determinant 
asociated with $B \in \cB_p(\cH)$, $p \in \bbN$, $p\geq 2$. 

Linear operators in the Hilbert space $L^2(\bbR; dt; \cH)$, in short, $L^2(\bbR; \cH)$, will be denoted by calligraphic boldface symbols of the type $\bsT$, to distinguish them from operators $T$ in $\cH$. In particular, operators denoted by 
$\bsT$ in the Hilbert space $L^2(\bbR;\cH)$ typically represent operators associated with a 
family of operators $\{T(t)\}_{t\in\bbR}$ in $\cH$, defined by
\begin{align}
&(\bsT f)(t) = T(t) f(t) \, \text{ for a.e.\ $t\in\bbR$,}    \no \\
& f \in \dom(\bsT) = \bigg\{g \in L^2(\bbR;\cH) \,\bigg|\,
g(t)\in \dom(T(t)) \text{ for a.e.\ } t\in\bbR;    \lb{1.S}  \\
& \quad t \mapsto T(t)g(t) \text{ is (weakly) measurable;} \, 
\int_{\bbR} dt \, \|T(t) g(t)\|_{\cH}^2 <  \infty\bigg\}.   \no
\end{align}
In the special case, where $\{T(t)\}$ is a family of bounded operators on $\cH$ with 
$\sup_{t\in\bbR}\|T(t)\|_{\cB(\cH)}<\infty$, the associated operator $\bsT$ is a bounded operator on $L^2(\bbR;\cH)$ with $\|\bsT\|_{\cB(L^2(\bbR;\cH))} = \sup_{t\in\bbR}\|T(t)\|_{\cB(\cH)}$.

For brevity we will abbreviate $I = I_{L^2(\bbR; dx)}$ and $\bsI = I_{L^2(\bbR; L^2(\bbR; dx))}$. Moreover, to simplify notation, we will frequently 
omit Lebesgue measure whenever possible and simply use $L^p(\bbR)$ instead of 
$L^p(\bbR; dx)$, and $L^p(\bbR^2)$ instead of $L^p(\bbR^2; dtdx)$, $p \geq 1$, etc. 

Rather than writing $M_{\psi}$ for the operator of multiplication by the (locally integrable) 
function $\psi$, we will abuse notation a bit and use the symbol $\psi$ in place of $M_{\psi}$. 

The symbol $AC_{\loc}(\bbR)$ represents locally absolutely continuous functions 
on $\bbR$. 

The open complex upper and lower half-planes are abbreviated by 
$\bbC_{\pm} = \{z\in\bbC \,|\, \Im(z) \gtrless 0\}$, respectively.

\section{The Spectral Shift Function  $\xi(\, \cdot \, ; A_+, A_-)$ And Its 
Approximating Sequence $\xi(\, \cdot \, ; A_{+,n}, A_-)$} \lb{s3}

In this section we will describe the spectral shift function  $\xi(\, \cdot \, ; A_+, A_-)$ 
and its approximating sequence, $\xi(\, \cdot \, ; A_{+,n}, A_-)$, $n \in \bbN$, in detail. 

We start with the basic assumptions used throughout this section: 

\begin{hypothesis} \lb{hB.0} 
Suppose the real-valued function $\phi$ satisfies   
\begin{equation}
\phi \in L^{\infty}(\bbR) \cap L^1(\bbR).  
\end{equation} 
\end{hypothesis}

Given Hypothesis \ref{hB.0}, we introduce the operators, 
\begin{equation}
A_- = - i \f{d}{dx}, \quad A_+ = - i \f{d}{dx} + \phi, \quad 
\dom(A_-) = \dom(A_+) = W^{1,2}(\bbR).    \lb{3.4} 
\end{equation}

Next, one writes (for fixed $z \in \bbC\backslash \bbR$), abbreviating $I = I_{L^2(\bbR)}$ for simplicity, 
\begin{align}
\begin{split} 
& (A_+ - z I)^{-1} - (A_- - z I)^{-1} =  - (A_- - z I)^{-1} \phi (A_+ - z I)^{-1}    \\
& \quad = - (A_- - z I)^{-1} \phi (A_- - z I)^{-1} 
\big[(A_- - z I) (A_+ - z I)^{-1}\big], \quad z \in \bbC \backslash \bbR. 
\end{split} 
\end{align}
Since (for $z \in \bbC \backslash \bbR$)
\begin{align}
& (A_- - z I) (A_+ - z I)^{-1}  \in \cB\big(L^2(\bbR)\big),    \\
& \big[(A_- - z I) (A_+ - z I)^{-1}\big]^{-1} = 
(A_+ - z I) (A_- - z I)^{-1} \in \cB\big(L^2(\bbR)\big), 
\end{align}
one concludes that given $p \in [1,\infty) \cup \{\infty\}$, $z \in \bbC \backslash \bbR$, 
\begin{align}
\begin{split} 
& \big[(A_+ - z I)^{-1} - (A_- - z I)^{-1}\big] \in \cB_p\big(L^2(\bbR)\big)   \\
& \quad \text{if and only if } \, 
(A_- - z I)^{-1} \phi (A_- - z I)^{-1} \in \cB_p\big(L^2(\bbR)\big).
\end{split} 
\end{align}
In particular, since $|\phi|^{1/2} \in L^2(\bbR)$ and 
$(|\cdot| - z)^{-1} \in L^2(\bbR; d\nu)$, an application of \cite[Theorem\ 4.1]{Si05} yields 
\begin{equation}
|\phi|^{1/2} (A_- - z I)^{-1} \in \cB_{2}\big(L^2(\bbR)\big), \quad 
z \in \bbC \backslash \bbR,     \lb{3.12} 
\end{equation}
and hence 
\begin{equation}
\big[(A_+ - z I)^{-1} - (A_- - z I)^{-1}\big] \in \cB_1\big(L^2(\bbR)\big), 
\quad z \in \bbC \backslash \bbR,     \lb{3.13} 
\end{equation}
upon decomposing $\phi$ into $\phi = |\phi|^{1/2} \sgn(\phi) |\phi|^{1/2}$. 

For later purpose it will be convenient to introduce the operator of multiplication by $\phi$ 
in $L^2(\bbR)$ and denote it by the symbol $B_+$. Thus, $B_+ \in \cB\big(L^2(\bbR)\big)$ and 
\begin{equation}
A_+ = A_- + B_+, \quad B_+ = \phi.    \lb{defB}
\end{equation} 

\begin{remark} \lb{r3.2} The fact \eqref{3.13} implies that the spectral shift function 
$\xi(\, \cdot \, ; A_+, A_-)$ for the pair $(A_+, A_-)$ exists and is well-defined up to an arbitrary 
additive real constant, satisfying 
\begin{equation} 
\xi(\, \cdot \, ; A_+, A_-) \in L^1\big(\bbR; (\nu^2 + 1)^{-1} d\nu\big).    \lb{3.15} 
\end{equation}  
In addition, the  trace formula, 
\begin{equation}
\tr_{L^2(\bbR; dx)}[f(A_+) - f(A_-)] = \int_{\bbR}  \xi(\nu; A_+, A_-)  d\nu \, f'(\nu),   \lb{3.16} 
\end{equation}
holds for a sufficiently wide class of functions $f$ (cf.\ \cite[Sect.\ 8.7]{Ya92}). 
Since $\sigma(A_{\pm}) = \bbR$, there is a priori no natural way to fix this open constant 
(although, possibilities to fix the constant via the unitary Cayley transforms of $A_{\pm}$ 
exist as discussed in \cite{BY93} and \cite[Sects.~8.5--8.7]{Ya92}). However, neither 
the integrability property \eqref{3.15}, nor the trace formula \eqref{3.16}, impose any restrictions 
on the arbitrary real constant inherent to the definition of $\xi(\, \cdot \, ; A_+, A_-)$. 
\hfill $\Diamond$
\end{remark}

Next, following Kato \cite[p.~30--31]{Ka80} or Yafaev \cite[p.~83--84]{Ya92}, we briefly 
describe scattering theory for the pair $(A_+, A_-)$:
One notes that $A_-$ generates translations, that is,
\begin{equation}
\big(e^{\pm i t A_-} u\big)(x) = u(x \pm t), \quad u \in L^2(\bbR),
\end{equation}
and introducing $U_+$, the unitary operator in $L^2(\bbR)$ of multiplication by 
\begin{equation}
U_+ = e^{- i \int_0^x dx' \phi(x')},     \lb{3.5} 
\end{equation}
one obtains 
\begin{equation}
e^{\pm i t A_+} = U_+ e^{\pm i t A_-} U_+^{-1},    \lb{3.5A} 
\end{equation}
and hence
\begin{align}
\big(e^{i t A_+} u\big)(x) &= e^{- i \int_0^x dx' \, \phi(x')} 
e^{i \int_0^{x+t} dx'  \phi(x')}u(x+t)    \no \\
&= e^{i \int_x^{x+t} dx' \, \phi(x')}u(x+t), \quad u \in L^2(\bbR).
\end{align}
Thus, introducing $\Omega(t; A_+, A_-) = e^{i t A_+} e^{- i t A_-}$, $t \in \bbR$, 
one obtains
\begin{equation}
(\Omega(t; A_+, A_-) u)(x) = e^{i \int_x^{x+t} dx' \, \phi(x')} u(x), 
\quad u \in L^2(\bbR), 
\end{equation}
and hence the wave operators are simply given by 
\begin{align}
\begin{split} 
(\Omega_{\pm}(A_+, A_-) u)(x)= (\slim_{t \to \pm \infty} \Omega(t; A_+, A_-) u)(x) 
= e^{i \int_{[x,\pm \infty)} dx' \, \phi(x')}u(x),& \\
u \in L^2(\bbR)&  
\end{split} 
\end{align}
(i.e., they are unitary operators in $L^2(\bbR)$ acting as operators of 
multiplication by a unimodular exponential). In addition, the unitary scattering operator 
in $L^2(\bbR)$ is then of the type 
\begin{equation}
S(A_+, A_-) = \Omega_+ (A_+, A_-)^* \Omega_- (A_+, A_-) 
= e^{- i \int_{\bbR} dx \, \phi(x)},
\end{equation}
that is, it acts as an operator of multiplication by a unimodular constant.

By general principles, $S(A_+, A_-)$ commutes with $A_-$ and hence 
decomposing
\begin{equation}
L^2(\bbR; dx) = \cF L^2(\bbR; d\nu) = \cF \int_{\bbR}^{\oplus} d\nu \, \bbC,  
\end{equation}
$\cF$ the Fourier transform, $S(A_+, A_-)$ in $L^2(\bbR)$ decomposes as 
\begin{equation}
S(A_+, A_-) = \cF \bigg(\int_{\bbR}^{\oplus} d\nu \, S(\nu; A_+, A_-) \bigg) \cF^{-1}, 
\end{equation}
where the reduced unitary scattering operator $S(\nu; A_+, A_-)$ (also known as the scattering matrix) in the one-dimensional Hilbert space $\bbC$ is just the constant (i.e., $\nu$-independent) 
phase factor,
\begin{equation}
S(\nu; A_+, A_-) = e^{- i \int_{\bbR} dx \, \phi(x)}, \quad \nu \in \bbR. 
\end{equation} 
By the Birman--Krein formula relating the determinant of the reduced scattering operator with the spectral shift function (cf.\ \cite{BK62}, \cite[Theorem\ 8.7.2]{Ya92}), one finally obtains
\begin{equation}
S(\nu; A_+, A_-) = e^{- i \int_{\bbR} dx \, \phi(x)} = e^{- 2 \pi i \wti \xi(\nu; A_+, A_-)}  
\, \text{ for a.e.\  $\nu \in \bbR$,}     \lb{3.27} 
\end{equation} 
and hence 
\begin{equation}
\wti \xi(\nu; A_+, A_-) = \f{1}{2 \pi} \int_{\bbR} dx \, \phi(x) + N(\nu) 
\, \text{ for a.e.\  $\nu \in \bbR$,}    \lb{3.28} 
\end{equation}
where $N(\cdot)$ is integer-valued. Here $\wti \xi(\, \cdot \, ; A_+,A_-)$ denotes a particularly 
normalized spectral shift function for the pair $(A_+, A_-)$, that is, a certain choice of the open 
additive (real) constant inherent to the definition of spectral shift functions (cf., also 
Remark \ref{r3.2}) has been made in $\wti \xi(\, \cdot \, ; A_+,A_-)$. Later in this section 
(cf.\ \eqref{B.21}) we will show 
that $\xi(\, \cdot \,; A_+, A_-)$ (and hence any spectral shift function 
associated with the pair $(A_+,A_-)$) is constant (and hence, continuous) on $\bbR$ and thus 
\begin{align} 
\xi(\, \cdot \,; A_+, A_-) &= \wti \xi(\, \cdot \,; A_+, A_-) + n_0    \\
&= \f{1}{2 \pi} \int_{\bbR} dx \, \phi(x) + n_1 \, \text{ for some $n_0, n_1 \in \bbZ$,}
\end{align}
is now defined up to an arbitrary additive integer (in accordance with \cite[eqs.\ (6.3), (8.6), Theorem\ 6.2]{BY93} and the paragraph following eq. (8.7), \cite[p.\ 286, 299, eq.\ (4), p.\ 297]{Ya92}). 

\begin{remark} \lb{r3.3} 
As discussed in \cite[Sects.~6--8]{BY93} and in great detail in \cite[Sects.~8.5--8.7]{Ya92}, culminating in \cite[Theorem~8.7.2]{Ya92}, the Birman--Krein formula 
\begin{equation}
S(\nu; A_+, A_-) = e^{- 2 \pi i \wti \xi(\nu; A_+, A_-)}\, \text{ for a.e.\  $\nu \in \bbR$,} 
\end{equation} 
is valid under condition \eqref{3.13} for a particular choice of the open additive constant in the 
spectral shift function $\xi(\, \cdot \,; A_+, A_-)$, leading to what we denoted by 
$\wti \xi(\, \cdot \,; A_+, A_-)$. The latter can be defined directly via the unitary Cayley transform 
of $A_{\pm}$ as described in equation (4) of \cite[Sect.~8.7]{Ya92} (see also 
\cite[Sects.~6--8]{BY93}), but also this definition fixes $\wti \xi(\, \cdot \,; A_+, A_-)$ only up 
to an integer as will be discussed in some detail in Appendix \ref{sA}. For additional discussions 
addressing the open integer in $\wti \xi(\, \cdot \,; A_+, A_-)$ we refer to \cite{Pu10}. 
\hfill $\Diamond$
\end{remark} 

Next, we apply Theorem \ref{tB.1} to the pair $(A_+, A_-)$, identifying $A_0$ 
with $A_-$, $A$ with $A_+$, and $B$ with the operator of multiplication by $\phi$, 
assuming again Hypothesis \ref{hB.0}. We refer to Appendix \ref{sA} for the definition 
and properties of modified Fredholm determinants ${\det}_{2,\cH}(I_{\cH} - A)$, 
$A \in \cB_2(\cH)$. 

For the resolvent of $A_-$ one computes 
\begin{align}
& \big((A_- - z I)^{-1} f\big)(x) = \begin{cases}
i \int_{-\infty}^x dx' \, e^{iz (x - x')} f(x'), & \Im(z) > 0,     \lb{B.9} \\
-i \int_x^{\infty} dx' e^{iz (x - x')} f(x'), & \Im(z) < 0, 
\end{cases} \\
& \hspace*{6cm} x \in \bbR, \; f \in L^2(\bbR; dx),   \no 
\end{align}
and hence the Green's function of $A_-$ is given by
\begin{equation}
(A_- - z I)^{-1}(x,x') = \begin{cases} i e^{iz (x - x')} \theta(x - x'), & \Im(z) > 0, \\
- i e^{iz (x - x')} \theta(x' - x), & \Im(z) < 0,
\end{cases} \quad x, x' \in \bbR,
\end{equation}
where 
\begin{equation} 
\theta(x) = \begin{cases} 1, & x > 0, \\ 0, & x < 0. \end{cases}
\end{equation} 
Thus, the integral kernel of $(A_- - z I)^{-1} \phi (A_- - z I)^{-1}$ is of the form
\begin{align}
& \big((A_- - z I)^{-1} \phi (A_- - z I)^{-1}\big)(x,x')    \no \\
& \quad = \int_{\bbR} dx'' \, (A_- - z I)^{-1}(x,x'') \phi(x'') (A_- - z I)^{-1}(x'',x')  \no \\
& \quad = - e^{iz(x - x')} \int_{\bbR} dx'' \, \theta(x - x'') \phi(x'') \theta(x'' - x'), 
\quad \Im(z) > 0, \; x, x' \in \bbR, 
\end{align}
and analogously for $\Im(z) < 0$. Hence, one concludes
\begin{equation}
\big((A_- - z I)^{-1} \phi (A_- - z I)^{-1}\big)(x,x) = 0, \quad 
z \in \bbC \backslash \bbR, \; x \in \bbR.    \lb{B.12} 
\end{equation}
By \eqref{3.12},
\begin{align}
& (A_- - z I)^{-1} \phi (A_- - z I)^{-1}     \\
& \quad = \big[(A_- - z I)^{-1} |\phi|^{1/2} \sgn(\phi)\big]
\big[|\phi|^{1/2} (A_- - z I)^{-1}\big] \in \cB_1\big(L^2(\bbR)\big), \quad 
z \in \bbC \backslash \bbR,   \no 
\end{align}
and thus, 
\begin{align}
\begin{split} 
\eta'(z) &= {\tr}_{L^2(\bbR)} \big((A_- - z I)^{-1} \phi (A_- - z I)^{-1}\big)    \\ 
&= \int_{\bbR} dx \, \big((A_- -z I)^{-1} \phi (A_- - z I)^{-1}\big)(x,x) = 0, 
\quad z \in \bbC \backslash \bbR,     \lb{B.13} 
\end{split} 
\end{align}
combining Examples V.2.19 and X.1.18 of \cite{Ka80}. Thus,
\begin{equation}
\eta (z) = \eta_{\pm}, \quad z \in \bbC_{\pm},
\end{equation}
for some constants $\eta_{\pm} \in \bbC$, implying 
\begin{equation}
\eta (\nu \pm i 0) = \eta_{\pm}, \quad \nu \in \bbR.    \lb{B.13a} 
\end{equation} 
In complete analogy to \eqref{B.9}--\eqref{B.13} one also obtains
\begin{align}
& \big((A_+ - z I)^{-1} f\big)(x) = \begin{cases}
i \int_{-\infty}^x dx' \, e^{iz (x - x')} e^{-i \int_{x'}^{x} dx'' \, \phi(x'')} f(x'), & \Im(z) > 0,  \\
-i \int_x^{\infty} dx' e^{iz (x - x')} e^{-i \int_{x'}^{x} dx'' \, \phi(x'')} f(x'), & \Im(z) < 0, 
\end{cases}  \no \\
& \hspace*{7.7cm} x \in \bbR, \; f \in L^2(\bbR; dx),   \\
& (A_+ - z I)^{-1}(x,x') = \begin{cases} i e^{iz (x - x')} e^{-i \int_{x'}^{x} dx'' \, \phi(x'')} \theta(x - x'), & \Im(z) > 0, \\
- i e^{iz (x - x')} e^{-i \int_{x'}^{x} dx'' \, \phi(x'')}\theta(x' - x), & \Im(z) < 0,
\end{cases}    \\
& \hspace*{9.15cm} \quad x, x' \in \bbR,  \no \\
& \big((A_+ - z I)^{-1} \phi (A_- - z I)^{-1}\big)(x,x')    \no \\
& \quad = \int_{\bbR} dx'' \, (A_+ - z I)^{-1}(x,x'') \phi(x'') (A_- - z I)^{-1}(x'',x')  \no \\
& \quad = - e^{iz(x - x')} \begin{cases} \int_{\bbR} dx'' \, e^{-i \int_{x''}^{x} dx''' \, \phi(x''')}
\theta(x - x'') \phi(x'') \theta(x'' - x'), & \Im(z) > 0, \\ 
\int_{\bbR} dx'' \, e^{-i \int_{x''}^{x} dx''' \, \phi(x''')}
\theta(x'' - x) \phi(x'') \theta(x' - x''), & \Im(z) < 0,
\end{cases}   \no \\
& \hspace*{9.4cm} x, x' \in \bbR,    
\end{align}
and hence, 
\begin{equation}
\big((A_+ - z I)^{-1} \phi (A_- - z I)^{-1}\big)(x,x) = 0, \quad 
z \in \bbC \backslash \bbR, \; x \in \bbR.  
\end{equation} 
Consequently,
\begin{align}
0 &= \tr_{L^2(\bbR)} \big((A_+ - z I)^{-1} - (A_- - z I)^{-1}\big)   \no  \\ 
&= - \tr_{L^2(\bbR)} \big((A_+ - z I)^{-1} \phi (A_- - z I)^{-1}\big)     \lb{3.30} \\ 
&= - \int_{\bbR} \f{\xi(\nu; A_+, A_-) d\nu}{(\nu - z)^2}, \quad z \in \bbC \backslash \bbR,   \no 
\end{align}
for some spectral shift function $\xi(\, \cdot \,; A_+, A_-)$ for the pair $(A_+, A_-)$ (all others 
differing from $\xi(\, \cdot \,; A_+, A_-)$ at most by a constant). Thus,
\begin{equation}
0 = \f{d}{dz} \int_{\bbR} \xi(\nu; A_+, A_-) d\nu \, \big[(\nu -z)^{-1} 
- \nu (\nu^2 + 1)^{-1}\big], \quad z \in \bbC \backslash \bbR,
\end{equation}
implies 
\begin{equation}
\int_{\bbR} \xi(\nu; A_+, A_-) d\nu \, \big[(\nu -z)^{-1} 
- \nu (\nu^2 + 1)^{-1}\big] = C, \quad \Im(z) > 0 
\end{equation}
for some constant $C \in \ol{\bbC_+}$. Equivalently, 
\begin{equation}
- \Re(z) + \int_{\bbR} \xi(\nu; A_+, A_-) d\nu \, \big[(\nu -z)^{-1} 
- \nu (\nu^2 + 1)^{-1}\big] = i \Im(C), \quad \Im(z) > 0 
\end{equation}
and hence the Stieltjes inversion formula (cf.\ \cite{AD56}) yields
\begin{equation}
\xi(\nu; A_+, A_-) = \pi^{-1} \Im(C)   \lb{B.21} 
\end{equation}
In particular, $\xi(\, \cdot \, ; A_+, A_-)$ has a constant (and hence, continuous) representative. 
Together with \eqref{3.27} and \eqref{3.28} this finally yields
\begin{equation}
\xi(\nu; A_+, A_-) = \f{1}{2 \pi} \int_{\bbR} dx \, \phi(x) + n_1, \quad \nu \in \bbR,   \lb{B.22}
\end{equation} 
for some $n_1 \in \bbZ$. The integer $n_1$ is unspecified at the moment, but the choice 
$n_1=0$ will naturally evolve as the result of an approximation procedure near the end of this section.

Next, we return to \eqref{B.12}--\eqref{B.13a} and prepare some facts that permit us to apply Theorem \ref{tB.1} to $\xi(\, \cdot \,; A_+,A_-)$. We start by noting that the integral kernel of 
$\sgn(\phi) |\phi|^{1/2} (A_- - z I)^{-1} |\phi|^{1/2}$ reads 
\begin{align}
& \sgn(\phi) |\phi|^{1/2} (A_- - z I)^{-1} |\phi|^{1/2} (x,x')    \no \\
& \quad = \pm 
i \sgn(\phi(x)) |\phi(x)|^{1/2} e^{\pm i z (x - x')} \theta(\pm(x - x')) |\phi(x')|^{1/2},  \lb{B.14} \\
& \hspace*{5.35cm} \pm \Im(z) > 0, \; x,x' \in \bbR,    \no 
\end{align}
and hence 
\begin{align}
& \big\|\sgn(\phi) |\phi|^{1/2} (A_- - z I)^{-1} |\phi|^{1/2}\big\|_{\cB_2(L^2(\bbR))}^2 
\no \\
& \quad = \pm \int_{\bbR} dx \, |\phi(x)| \int_{\mp \infty}^x dx' \, e^{\mp 2 \Im(z) (x - x')} |\phi(x')| 
\lb{B.15} \\
& \quad \leq \|\phi\|_{L^1(\bbR)}^2, \quad \pm \Im(z) > 0.     \lb{B.16} 
\end{align}
Since $e^{\mp 2 \Im(z) (x - x')} \leq 1$, $\pm \Im(z) \geq 0$, $\pm (x - x') \geq 0$, Lebesgue's 
dominated convergence theorem yields   
\begin{equation}
\lim_{z \to \pm i \infty} \big\|\sgn(\phi) |\phi|^{1/2} (A_- - z I)^{-1} 
|\phi|^{1/2}\big\|_{\cB_2(L^2(\bbR))} = 0.    \lb{B.17} 
\end{equation}

In addition, the same arguments yield 
\begin{align}
& \big\||\phi|^{1/2} (A_- - z I)^{-1} \big\|_{\cB_2(L^2(\bbR))}^2 
\no \\
& \quad = \pm \int_{\bbR} dx \, |\phi(x)| \int_{\mp \infty}^x dx' \, e^{\mp 2 \Im(z) (x - x')} 
\lb{B.15A} \\
& \quad = [\pm 2 \Im(z)]^{-1} \|\phi\|_{L^1(\bbR)}, \quad \pm \Im(z) > 0.     \lb{B.16A} 
\end{align}

In fact, equation \eqref{B.14}--\eqref{B.16}, together with their counterparts for $\Im(z) < 0$,  
yield more as they also prove the existence of the limits 
\begin{align}
& \lim_{\varepsilon \downarrow 0} 
\sgn(\phi) |\phi|^{1/2} (A_- - (\nu \pm i \varepsilon)I)^{-1} |\phi|^{1/2}     \lb{B.16B} \\
& \quad : = \sgn(\phi) |\phi|^{1/2} (A_- - (\nu \pm i 0)I)^{-1} |\phi|^{1/2}
\, \text{ in $\cB_2\big(L^2(\bbR)\big)$-norm,} \quad \nu \in \bbR,   \no 
\end{align} 
as well as,
\begin{align}
& \big\|\sgn(\phi) |\phi|^{1/2} (A_- - (\nu \pm i 0) I)^{-1} |\phi|^{1/2}\big\|_{\cB_2(L^2(\bbR))}^2 
\no \\ 
& \quad \leq \|\phi\|_{L^1(\bbR)}^2, \quad \nu \in \bbR.     \lb{B.16a} 
\end{align}

By \eqref{B.9}, the limits $\sgn(\phi) |\phi|^{1/2} (A_- - (\nu + i 0)I)^{-1} |\phi|^{1/2}$ and  
$\sgn(\phi) |\phi|^{1/2} (A_- - (\nu - i 0)I)^{-1} |\phi|^{1/2}$ differ. 

Thus, Theorem \ref{tB.1} applies, and combining \eqref{B.7} and \eqref{B.13a} yields 
for some constant $d_0 \in \bbR$, 
\begin{align}
& \xi(\nu; A_+, A_-)   \no \\
& \quad = \pi^{-1} \Im\big(\ln\big({\det}_{2, L^2(\bbR)}
\big(I + \sgn(\phi) |\phi|^{1/2} (A_- - (\nu + i 0) I)^{-1} |\phi|^{1/2}\big)\big)\big)
+ d_0     \no \\
& \hspace*{9cm} \text{for a.e.\ $\nu \in \bbR$.}      \lb{B.13b} 
\end{align}

Morever, \eqref{B.14}, \eqref{B.15}, \eqref{B.16B}, and \eqref{B.16a} once more combined 
with Lebesgue's dominated convergence theorem also prove that the map 
\begin{align}
\bbR \ni \nu \mapsto \sgn(\phi) |\phi|^{1/2} (A_- - (\nu + i 0)I)^{-1} |\phi|^{1/2} 
\in \cB_2\big(L^2(\bbR)\big)     \lb{B.18} 
\end{align} 
is continuous and uniformly bounded with respect to $\nu \in \bbR$ in the 
$\cB_2\big(L^2(\bbR)\big)$-norm. Consequently, also the map 
\begin{align}
\bbR \ni \nu \mapsto {\det}_{2,L^2(\bbR)}
\big(I + \sgn(\phi) |\phi|^{1/2} (A_- - (\nu + i 0)I)^{-1} |\phi|^{1/2}\big)    \lb{B.19} 
\end{align} 
is continuous on $\bbR$, employing \eqref{detcont}. 

Finally, before turning to approximations, we claim that 
\begin{equation} 
{\det}_{2, L^2(\bbR)}
\big(I + \sgn(\phi) |\phi|^{1/2} (A_- - (\nu \pm i 0) I)^{-1} |\phi|^{1/2}
\big) = 1, \quad \nu \in \bbR.    \lb{B.49} 
\end{equation}
Indeed, combining \eqref{B.8}, \eqref{B.13}, and  \eqref{3.30} results in 
\begin{equation}
0 = \f{d}{dz} \ln\big({\det}_{2,L^2(\bbR)}\big(I + \sgn(\phi) |\phi|^{1/2} 
(A_- - z I)^{-1} |\phi|^{1/2}\big)\big), \quad z \in \bbC \backslash \bbR. 
\end{equation}
Thus, analyticity of 
${\det}_{2,L^2(\bbR)}\big(I + \sgn(\phi) |\phi|^{1/2} (A_- - \cdot \, I)^{-1} |\phi|^{1/2}\big)$ 
on $\bbC_{\pm}$ (e.g., as a consequence of \eqref{detcont}) yields 
\begin{equation}
{\det}_{2,L^2(\bbR)}\big(I + \sgn(\phi) |\phi|^{1/2} (A_- - z I)^{-1} |\phi|^{1/2}\big) = C_{\pm}, 
\quad z \in \bbC_{\pm}  
\end{equation}
for some constants $C_{\pm} \in \bbC$. By \eqref{detcont} and \eqref{B.17}, and by applying continuity of $\sgn(\phi) |\phi|^{1/2} (A_- - (\nu \pm i \varepsilon) I)^{-1} |\phi|^{1/2}$ in 
$\cB_2\big(L^2(\bbR)\big)$-norm as $\varepsilon \downarrow 0$, $C_{\pm} = 1$, proving 
\eqref{B.49}. 

As a preparation in connection with approximations to be studied in the remainder of this 
section, we first recall the following standard convergence property for trace ideals:

\begin{lemma}\lb{l3.6}
Let $p\in[1,\infty)$ and assume that $R,R_n,T,T_n\in\cB(\cH)$, 
$n\in\bbN$, satisfy
$\slim_{n\to\infty}R_n = R$  and $\slim_{n\to \infty}T_n = T$ and that
$S,S_n\in\cB_p(\cH)$, $n\in\bbN$, satisfy 
$\lim_{n\to\infty}\|S_n-S\|_{\cB_p(\cH)}=0$.
Then $\lim_{n\to\infty}\|R_n S_n T_n^\ast - R S T^\ast\|_{\cB_p(\cH)}=0$.
\end{lemma}

To set up approximations for $A_+$, we now deviate from the usual approximation 
procedure originally employed in \cite{GLMST11} and \cite{Pu08}: We introduce 
\begin{equation}\label{def_chi_n}
\chi_n(\nu) = \f{n}{(\nu^2 + n^2)^{1/2}}, \quad \nu \in \bbR, \; n \in \bbN, 
\end{equation} 
and hence obtain  
\begin{equation}
\slim_{n \to \infty} \chi_n(A_-) = I,     \lb{slim} 
\end{equation}
by an elementary application of the spectral theorem for $A_-$. The precise form 
of $\chi_n$ is of course immaterial, we just need property \eqref{slim} (and the 
Hilbert--Schmidt property \eqref{2.61} below). (One notes, however, that 
$\|\chi_n(A_-) - I\|_{\cB(L^2(\bbR))} = \sup_{\nu \in \bbR}\big|(\nu^2 + n^2)^{-1/2}n -1\big|= 1$, 
$n \in \bbN$, so convergence in the strong operator topology in \eqref{slim} is essential.) 
We recall our convention to denote the operator of multiplication by $\phi$ in $L^2(\bbR)$
by the symbol $B_+$ (cf.\ also \eqref{defB}), and introduce 
\begin{align}
& B_{+,n} =  \chi_n(A_-) B_+ \chi_n(A_-) = \chi_n(A_-) \phi \chi_n(A_-), \quad n \in \bbN, 
\lb{Bn}  \\
& A_{-,n} = A_-, \quad \dom(A_{-,n}) = \dom(A_-),  \quad n \in \bbN,     \lb{A-n} \\ 
& A_{+,n} = A_- + B_{+,n}, \quad \dom(A_{+,n}) = \dom(A_-),  \quad n \in \bbN,   \lb{A+n} 
\end{align}
and conclude 
\begin{equation}
B_{+,n} = \chi_n(A_-) \phi \chi_n(A_-) \in \cB_1\big(L^2(\bbR)\big), \quad n \in \bbN.    \lb{2.58} \\
\end{equation}
Here we used that
\begin{equation}
\big\| |\phi|^{1/2} \chi_n(A_-)\big\|_{\cB_2(L^2(\bbR))} = (2 \pi)^{-1/2}
\big\| |\phi|^{1/2}\big\|_{L^2(\bbR)} \|\chi_n\|_{L^2(\bbR; d\nu)} < \infty     \lb{2.61} 
\end{equation}
(cf.\ the proof of \cite[Theorem\ 4.1]{Si05} on p.~39). In particular, similarly to \eqref{3.13} one has 
\begin{equation}
\big[(A_{+,n} - z I)^{-1} - (A_- - z I)^{-1}\big] \in \cB_1 \big(L^2(\bbR)\big), 
\quad n \in \bbN, \; z \in \bbC \backslash \bbR,    \lb{AB1n}
\end{equation} 
In addition, we note that by 
\cite[Theorem~4.1]{Si05} (for $z \in \bbC \backslash \bbR$) and Lemma \ref{l3.6},   
\begin{align}
& \lim_{n\to\infty} \big\|(A_- - z I)^{-1/2} \chi_n(A_-) |\phi|^{1/2} 
- (A_- - z I)^{-1/2} |\phi|^{1/2}\big\|_{\cB_4(L^2(\bbR))} = 0,     \\
& \lim_{n\to\infty} \big\|\sgn(\phi) |\phi|^{1/2} (A_- - z I)^{-1/2}  \chi_n(A_-)    \no \\
& \hspace*{1cm}  - \sgn(\phi) |\phi|^{1/2} (A_- - z I)^{-1/2}\big\|_{\cB_4(L^2(\bbR))} = 0, \\
& \lim_{n \to \infty} \big\|B_{+,n} (A_- - z I)^{-1} - B_+ (A_- -z I)^{-1}\big\|_{\cB_2(L^2(\bbR))} =0, 
\lb{B2Bn} 
\end{align}
and 
\begin{align} 
\begin{split} 
& \lim_{n\to\infty} \big\|\sgn(\phi) |\phi|^{1/2} \chi_n(A_-) (A_- - z I)^{-1/2} 
\chi_n(A_-) |\phi|^{1/2}    \lb{B.23} \\
& \hspace*{1cm} - \sgn(\phi) |\phi|^{1/2} (A_- - z I)^{-1/2} |\phi|^{1/2}\big\|_{\cB_2(L^2(\bbR))} = 0, \quad z \in \bbC \backslash \bbR.  
\end{split} 
\end{align}
Moreover, employing 
\begin{align} 
& (A_- - z I)^{-1} B_{+,n} (A_- - z' I)^{-1} 
= \chi_n(A_-) (A_- - z I)^{-1} B_+ (A_- - z' I)^{-1} \chi_n(A_-),    \no \\ 
& \hspace*{8cm} n \in \bbN, \; z, z' \in \bbC \backslash \bbR,   
\end{align} 
Lemma \ref{l3.6} also implies
\begin{align}
& \lim_{n \to \infty} \big\|(A_- - z I)^{-1} B_{+,n} (A_- - z' I)^{-1} - 
(A_- - z I)^{-1} B_+ (A_- -z' I)^{-1}\big\|_{\cB_1(L^2(\bbR))}      \no \\
& \quad =0, \quad z, z' \in \bbC \backslash \bbR.     \lb{B1Bn}
\end{align}
Relations \eqref{B2Bn} and \eqref{B1Bn} will be used in Section \ref{s5}. 

An application of \eqref{B.6} to the pairs $(A_+, A_-)$ and $(A_{+,n}, A_-)$, 
taking into account \eqref{B.13}, thus yields 
\begin{align}
& \int_{\bbR} \xi(\nu; A_+, A_-) d\nu  
\bigg(\f{1}{\nu - z} - \f{1}{\nu - z_0}\bigg)    \no \\
& \quad  
= \ln\bigg(\f{{\det}_{2, L^2(\bbR)}\big(I + \sgn(\phi) |\phi|^{1/2} 
(A_- - z I)^{-1} |\phi|^{1/2}\big)}
{{\det}_{2, L^2(\bbR)}\big(I + \sgn(\phi) |\phi|^{1/2} 
(A_- - z_0 I)^{-1} |\phi|^{1/2}\big)}\bigg) = 0,       \lb{B.24} \\
& \hspace*{5.2cm} z, z_0 \in \bbC \backslash \bbR, \; \Im(z) \Im(z_0) > 0,  
\end{align}
since $\xi(\, \cdot \,; A_+, A_-)$ is constant by \eqref{B.22}, and 
\begin{align}
& \int_{\bbR} \xi(\nu; A_{+,n}, A_-) d\nu 
\bigg(\f{1}{\nu - z} - \f{1}{\nu - z_0}\bigg) = \eta_n (z) - \eta_n (z_0)   \no \\
& \quad  
+ \ln\bigg(\f{{\det}_{2, L^2(\bbR)}\big(I + \sgn(\phi) |\phi|^{1/2} 
\chi_n(A_-) (A_- - z I)^{-1} \chi_n(A_-) |\phi|^{1/2}\big)}
{{\det}_{2, L^2(\bbR)}\big(I + \sgn(\phi) |\phi|^{1/2}
\chi_n(A_-) (A_- - z_0 I)^{-1} \chi_n(A_-) |\phi|^{1/2}\big)}
\bigg),     \no \\ 
& \hspace*{5cm} n \in \bbN, \; z, z_0 \in \bbC \backslash \bbR, 
\; \Im(z) \Im(z_0) > 0.     \lb{B.25}
\end{align} 
Here
\begin{align}
\eta_n^{\prime} (z) &=  {\tr}_{L^2(\bbR)} \big((A_- - z I)^{-1} \chi_n(A_-) 
\phi \chi_n(A_-) (A_- - z I)^{-1}\big)     \\ 
& = {\tr}_{L^2(\bbR)} \big(\chi_n(A_-)(A_- - z I)^{-1}  
\phi (A_- - z I)^{-1}\chi_n(A_-)\big), 
\quad z \in \bbC \backslash \bbR, \; n \in \bbN,   \no 
\end{align}
and hence by appealing to the fact that 
$(A_- -z I)^{-1} \phi (A_- -z I)^{-1} \in \cB_1\big(L^2(\bbR)\big)$, and using 
Lemma \ref{l3.6} and \eqref{B.13}, one obtains 
\begin{equation}
\lim_{n \to \infty} \eta_n^{\prime} (z) = \eta'(z) = 0, \quad z \in \bbC \backslash \bbR. 
\end{equation}
Moreover, by \eqref{B.49} and \eqref{B.23} also 
\begin{align}
\begin{split}
& \lim_{n \to \infty} {\det}_{2, L^2(\bbR)}\big(I + \sgn(\phi) |\phi|^{1/2} 
\chi_n(A_-) (A_- - z I)^{-1} \chi_n(A_-) |\phi|^{1/2}\big)    \lb{B.29} \\
& \quad = {\det}_{2, L^2(\bbR)}\big(I + \sgn(\phi) |\phi|^{1/2} 
(A_- - z I)^{-1} |\phi|^{1/2}\big) =1, \quad z \in \bbC \backslash \bbR, 
\end{split}
\end{align}
and hence combining \eqref{B.24}--\eqref{B.29}, we obtain in passing,  
\begin{align}
& \lim_{n \to \infty} \int_{\bbR} \xi(\nu; A_{+,n}, A_-) d\nu 
\bigg(\f{1}{\nu - z} - \f{1}{\nu - z_0}\bigg)      \lb{B.30} \\
& \quad = \int_{\bbR} \xi(\nu; A_+, A_-) d\nu  
\bigg(\f{1}{\nu - z} - \f{1}{\nu - z_0}\bigg) = 0, \quad 
z, z_0 \in \bbC \backslash \bbR, \; \Im(z) \Im(z_0) > 0,   \no 
\end{align}
again, since $\xi(\, \cdot \,; A_+, A_-)$ is constant by \eqref{B.22}. 

Next, upon investigating the explicit integral kernels, one infers that the   
analysis in \eqref{B.14}--\eqref{B.19} in connection with the 
operator $\sgn(\phi) |\phi|^{1/2} (A_- - (\cdot + i 0) I)^{-1} |\phi|^{1/2}$
now also applies to 
$\sgn(\phi) |\phi|^{1/2} \chi_n(A_-) (A_- - (\nu + i 0) I)^{-1}\chi_n(A_-) |\phi|^{1/2}$, 
in particular, for each $n \in \bbN$, the map 
\begin{align}
\bbR \ni \nu \mapsto \sgn(\phi) |\phi|^{1/2} \chi_n(A_-) (A_- - (\nu + i 0)I)^{-1} 
\chi_n(A_-) |\phi|^{1/2} 
\in \cB_2\big(L^2(\bbR)\big)     \lb{B.32} 
\end{align} 
is continuous and uniformly bounded with respect to $\nu \in \bbR$ in the 
$\cB_2\big(L^2(\bbR)\big)$-norm and hence for each $n \in \bbN$, also 
\begin{align}
\bbR \ni \nu \mapsto {\det}_{2,L^2(\bbR)}
\big(I + \sgn(\phi) |\phi|^{1/2} \chi_n(A_-) (A_- - (\nu + i 0)I)^{-1} \chi_n(A_-) 
|\phi|^{1/2}\big)    \lb{B.33} 
\end{align} 
is continuous on $\bbR$. 

At this point we turn to the computation of $\eta_n (\cdot)$, $n \in \bbN$. Given the facts 
\begin{equation}
(A_-^2 + \mu^2 I)^{-1}(x,x') = \f{1}{2 \mu} e^{- \mu |x - x'|}, \quad 
\mu > 0, \; x, x' \in \bbR,   \lb{B.34} 
\end{equation}
and \eqref{2.58}, employing cyclicity of the trace yields 
\begin{align}
\eta_n(z) &= \tr_{L^2(\bbR)} \big((A_- -z I)^{-1} \chi_n(A_-)^2 \phi\big) + d_n 
\no \\
&= \pm i n^2 \int_{\bbR} dx \int_{\bbR} dx' \, e^{\pm i z (x - x')} 
\chi_{(\mp \infty,x]}(x') \f{e^{-n|x' - x|}}{2n} \phi(x) + d_n     \no \\
&= \f{\pm i n}{2(n \mp iz)} \int_{\bbR} dx \, \phi(x) + d_n, \quad \pm\Im(z) \geq 0, \; n \in \bbN, 
\lb{B.35}
\end{align}
for some integration constants $d_n \in \bbC$, $n \in \bbN$. In particular, 
$\eta_n(\cdot + i0)$, $n \in \bbN$, is continuous on $\bbR$ (as required in 
Theorem \ref{tB.1}\,$(iii)$) and 
\begin{equation}
\Im(\eta_n(\nu + i0)) = \f{1}{2} \f{n^2}{\nu^2 + n^2} \int_{\bbR} dx \, \phi(x) + \Im(d_n), 
\quad \nu \in \bbR, \; n \in \bbN.     \lb{B.35a} 
\end{equation}
Hence \eqref{B.7} applies and yields 
\begin{align} 
& \xi(\nu; A_{+,n}, A_-) = \pi^{-1} \Im\big(\ln\big({\det}_{2, L^2(\bbR)}
\big(I + \sgn(\phi) |\phi|^{1/2} \chi_n(A_-) (A_- - (\nu + i 0) I)^{-1}    \no \\ 
& \quad \times \chi_n(A_-) |\phi|^{1/2}
\big)\big)\big) + \pi^{-1} \Im(\eta_n(\nu + i0)) + c_n  \,     
\text{ for a.e.\ } \nu \in \bbR, \; n \in \bbN,     \lb{B.36} 
\end{align} 
for some constants $c_n \in \bbR$, $n\in \bbN$. Because of \eqref{2.58}, 
$c_n$, $n \in \bbN$, are uniquely determined via the requirement  
\begin{equation}
\xi(\, \cdot \,; A_{+,n}, A_-) \in L^1(\bbR; d\nu), \quad n \in \bbN.     \lb{B.37}
\end{equation}

To study the asymptotic behavior of $\xi(\, \cdot \,; A_{+,n}, A_-)$, $n \in \bbN$, as 
$|\nu| \to \infty$, one observes that the integral kernel of 
$\sgn(\phi) |\phi|^{1/2} (A_- - z I)^{-1} \chi_n(A_-)^2 |\phi|^{1/2}$ is given by
\begin{align}
& \sgn(\phi) |\phi|^{1/2} (A_- - z I)^{-1} \chi_n(A_-)^2 |\phi|^{1/2}(x,x')     \no \\
& \quad = \pm (n/2) i \sgn(\phi(x)) |\phi(x)|^{1/2} \int_{\bbR} dx'' \, 
e^{\pm i z (x - x'')} \chi_{(\mp \infty,x]}(x'') e^{-n |x'' - x'|} |\phi(x')|^{1/2},    \no \\
& \hspace*{6.2cm} \pm \Im(z) \geq 0, \; x, x' \in \bbR, \; n \in \bbN.      \lb{B.38}
\end{align}
Here we employed \eqref{B.34} once again. 
With $z = \nu + i 0$, $\nu \in \bbR$, one thus computes
\begin{align}
& \big\|\sgn(\phi) |\phi|^{1/2} (A_- - (\nu + i0) I)^{-1} \chi_n(A_-)^2 
|\phi|^{1/2}\big\|_{\cB_2(L^2(\bbR))}^2     \no \\
& \quad = \f{n^2}{4} \int_{\bbR} dx \, |\phi(x)| \int_{\bbR} dx' \, |\phi(x')|
\bigg|\int_{-\infty}^x dx'' \, e^{-i \nu x''} e^{-n |x'' - x'|}\bigg|^2    \no \\ 
& \quad = \f{n^2}{4} \int_{\bbR} dx \, |\phi(x)| \bigg\{\int_{-\infty}^x dx' \, |\phi(x')| 
\bigg| \int_{-\infty}^{x'} dx'' \, e^{-i \nu x''} e^{-n(x' - x'')}    \no \\
& \hspace*{6cm} + \int_{x'}^x dx'' \, e^{-i \nu x''} e^{-n(x'' - x')}\bigg|^2     \no \\
& \hspace*{3.45cm} + \int_x^{\infty} dx' \, |\phi(x')| \bigg|\int_{-\infty}^x dx'' \, 
e^{-i \nu x''} e^{-n(x' - x'')}\bigg|^2\bigg\}     \no \\
& \quad \leq \f{5}{2} \f{n^2}{\nu^2 + n^2} \|\phi\|_{L^1(\bbR)}^2, \quad 
\nu \in \bbR, \; n \in \bbN.     \lb{B.39} 
\end{align}
Thus,
\begin{align}
& \ln\big({\det}_{2, L^2(\bbR)}
\big(I + \sgn(\phi) |\phi|^{1/2} \chi_n(A_-) (A_- - (\nu + i 0) I)^{-1} 
\chi_n(A_-) |\phi|^{1/2}\big)\big)    \no \\
& \quad \underset{|\nu| \to \infty}{=} \Oh\big(|\nu|^{-2}\big), \quad n \in \bbN,   \lb{B.41}
\end{align}
employing that for $\zeta \in \bbC$, with $|\zeta|$ sufficiently small, 
\begin{equation}
\ln({\det}_{2,\cH}(I_{\cH} + \zeta T)) = \sum_{m=2}^{\infty} \f{(-1)^{m+1}}{m} 
\zeta^m \tr_{\cH}(T^m), \quad T \in \cB_2(\cH)  
\end{equation}
(cf.\ \cite[p.~76]{Si05}). Combining \eqref{B.35a}, \eqref{B.36}, and \eqref{B.41} then yields 
\begin{equation}
\xi(\nu; A_{+,n},A_-) \underset{|\nu| \to \infty}{=} \Oh(|\nu|^{-2}) + \pi^{-1} \Im(d_n) + c_n , 
\quad n \in \bbN, 
\end{equation}
implying 
\begin{equation}
\pi^{-1} \Im(d_n) + c_n = 0, \quad n \in \bbN,    \lb{B.44} 
\end{equation} 
because of the integrability condition \eqref{B.37}. Thus, combining \eqref{B.35a}, 
\eqref{B.36}, and \eqref{B.44}, one obtains 
\begin{align} 
& \xi(\nu; A_{+,n}, A_-)  = \pi^{-1} \Im\big(\ln\big({\det}_{2, L^2(\bbR)}
\big(I + \sgn(\phi) |\phi|^{1/2} \chi_n(A_-) (A_- - (\nu + i 0) I)^{-1}    \no \\ 
& \quad \times \chi_n(A_-) |\phi|^{1/2}
\big)\big)\big) + \f{1}{2 \pi} \f{n^2}{\nu^2 + n^2} \int_{\bbR} dx \, \phi(x)  \,     
\text{ for a.e.\ } \nu \in \bbR, \; n \in \bbN.     \lb{B.45} 
\end{align} 

Next, we study the limit $n\to\infty$ of $\xi(\, \cdot \, ; A_{+,n}, A_-)$.

\begin{lemma} \lb{lB.6}
Assume Hypothesis \ref{hB.0}. Then for each $\nu \in \bbR$,
\begin{align}
\begin{split}
& \lim_{n \to \infty} \big\|\sgn(\phi) |\phi|^{1/2} \chi_n(A_-) (A_- - (\nu + i 0) I)^{-1} 
\chi_n(A_-)|\phi|^{1/2}     \lb{B.88} \\
& \hspace*{1cm} - \sgn(\phi) |\phi|^{1/2} (A_- - (\nu + i 0) I)^{-1} |\phi|^{1/2}
\big\|_{\cB_2(L^2(\bbR))} = 0, 
\end{split} 
\end{align}
and 
\begin{equation}
\lim_{n \to \infty} \xi(\nu; A_{+,n}, A_-) = \f{1}{2 \pi} \int_{\bbR} dx \, \phi(x), 
\quad \nu \in \bbR. 
\end{equation}
\end{lemma}
\begin{proof}
Employing the integral kernels \eqref{B.14} and \eqref{B.38} for $z = \nu + i0$, one 
estimates for each fixed $\nu \in \bbR$, 
\begin{align}
& \big\|\sgn(\phi) |\phi|^{1/2} \chi_n(A_-) (A_- - (\nu + i 0) I)^{-1} 
\chi_n(A_-)|\phi|^{1/2}     \\
& \;\; - \sgn(\phi) |\phi|^{1/2} (A_- - (\nu + i 0) I)^{-1} |\phi|^{1/2}
\big\|_{\cB_2(L^2(\bbR))}^2     \no \\
& \quad = \int_{\bbR} dx \, |\phi(x)| \int_{\bbR} dx' \, |\phi(x')| \bigg|\f{i n}{2} 
\int_{-\infty}^x dx'' \, e^{i \nu (x - x'')} e^{-n|x'' - x'|}    \no \\
& \hspace*{4.7cm} - i e^{i \nu (x - x')} \theta(x-x')\bigg|^2    \no \\
& \quad = \int_{\bbR} dx \, |\phi(x)| \int_{-\infty}^x dx' \, |\phi(x')| \bigg|\f{n}{2} 
\int_{-\infty}^x dx'' \, e^{- i \nu x''} e^{-n|x'' - x'|} - e^{- i \nu x'}\bigg|^2   \no \\
& \qquad + \int_{\bbR} dx \, |\phi(x)| \int_x^{\infty} dx' \, |\phi(x')| \bigg|\f{n}{2} 
\int_{-\infty}^x dx'' \, e^{- i \nu x''} e^{-n|x'' - x'|}\bigg|^2   \no \\
& \quad = \int_{\bbR} dx \, |\phi(x)| \int_{-\infty}^x dx' \, |\phi(x')| \bigg|\f{\nu^2}{n^2 + \nu^2} 
+ \f{n(n-i \nu)}{2(n^2 + \nu^2)} e^{- i \nu (x-x')} e^{- n(x-x')}\bigg|^2     \no \\
& \qquad + \int_{\bbR} dx \, |\phi(x)| \int_x^{\infty} dx' \, |\phi(x')| 
\f{n^2}{4(n^2 + \nu^2)} e^{- 2 n (x' - x)} \underset{n \to \infty}{\longrightarrow} 0,  \lb{B.91} 
\end{align}
proving \eqref{B.88}. Moreover, combining \eqref{B.49}, \eqref{B.35a}, \eqref{B.36}, \eqref{B.44}, 
and \eqref{B.88} yields   
\begin{align}
& \lim_{n \to \infty} \xi(\nu; A_{+,n}, A_-)    \no \\
& \quad = \pi^{-1} \Im\big(\ln\big({\det}_{2, L^2(\bbR)}
\big(I + \sgn(\phi) |\phi|^{1/2} (A_- - (\nu + i 0) I)^{-1} |\phi|^{1/2}
\big)\big)\big)    \no \\
& \qquad + \f{1}{2 \pi} \int_{\bbR} dx \, \phi(x)    \no \\
& \quad = \f{1}{2 \pi} \int_{\bbR} dx \, \phi(x), \quad \nu \in \bbR.    \lb{B.96} 
\end{align}
\end{proof}

To proceed, we need one additional resolvent approximation result:

\begin{lemma} \lb{lB.7}
Assume Hypothesis \ref{hB.0}. Then 
\begin{equation}
\lim_{n \to \infty} \big\|(A_{+,n} - z I)^{-1} - (A_+ - z I)^{-1}\big\|_{\cB_1(L^2(\bbR))} = 0,   
\quad z \in \bbC \backslash \bbR.      \lb{limB-1}
\end{equation} 
\end{lemma}
\begin{proof}
One writes
\begin{align}
& (A_{+,n} - z I)^{-1} - (A_+ - z I)^{-1} = \big[(A_{+,n} - z I)^{-1} - (A_- - z I)^{-1}\big]    \no \\
& \qquad - \big[(A_+ - z I)^{-1} - (A_- - z I)^{-1}\big]   \no \\
& \quad = - \chi_n(A_-) (A_- - z I)^{-1} \phi  (A_- - z I)^{-1} \chi_n(A_-) 
\big[(A_- - z I)  (A_{+,n} - z I)^{-1}\big]   \no \\
& \qquad +  (A_- - z I)^{-1} \phi  (A_- - z I)^{-1}
\big[ (A_- - z I)  (A_+ - z I)^{-1}\big]    \no \\
& \quad = - \chi_n(A_-) (A_- - z I)^{-1} \phi  (A_- - z I)^{-1} \chi_n(A_-)    \no \\
& \qquad \times \big[I - \chi_n(A_-) \phi \chi_n(A_-) (A_{+,n} - z I)^{-1}\big]   \no \\
& \qquad +  (A_- - z I)^{-1} \phi  (A_- - z I)^{-1}
\big[I - \phi (A_+ - z I)^{-1}\big], \quad z \in \bbC \backslash \bbR.  
\end{align}
Thus, relying on Lemma \ref{l3.6} and  \eqref{slim} once again, it suffices to prove that 
\begin{equation}
\slim_{n \to \infty} (A_{+,n} - z I)^{-1} = (A_+ - z I)^{-1}, \quad z \in \bbC \backslash \bbR, 
\end{equation}
but this immediately follows from
\begin{align}
& (A_{+,n} - z I)^{-1} = \big[I + (A_- - z I)^{-1} \chi_n(A_-) \phi \chi_n(A_-)\big]^{-1} (A_- - z I)^{-1}, \\
& (A_+ - z I)^{-1} = \big[I + (A_- - z I)^{-1} \phi\big]^{-1} (A_- - z I)^{-1}, 
\end{align}
employing the fact that strong convergence for a sequence of bounded operators is 
equivalent to strong resolvent convergence, initially, for $|\Im(z)|$ suffficiently large, and subsequently, for all $z \in \bbC \backslash \bbR$ by analytic continuation with respect to $z$. 
\end{proof} 

Finally, going beyond the approximation $A_{+,n}$ of $A_+$, we now introduce the 
following path $\{A_+(s)\}_{s \in [0,1]}$, where 
\begin{align}
& A_+(s) = A_- + \wti \chi_s (A_-) \phi \wti \chi_s (A_-), 
\quad \dom(A_+(s)) = \dom(A_-),   \quad s \in [0,1],    \lb{B.102} \\
& \wti \chi_s (\nu) = \big[(1-s) \nu^2 + 1\big]^{-1/2}, \quad \nu \in \bbR, \; s \in [0,1],
\end{align}
in particular, 
\begin{equation} 
A_+(0) = A_{+,1} \text{ (cf.\ \eqref{A+n} with $n=1$) and } \, A_+(1) = A_+. 
\end{equation} 
Moreover, in complete analogy to \eqref{limB-1}, the family 
$A_+(s)$ depends continuously on $s \in [0,1]$ with respect to the metric 
\begin{equation}
d(A,A') = \big\|(A - i I)^{-1} - (A' - i I)^{-1}\big\|_{\cB_1(L^2(\bbR))}    \lb{B.105}
\end{equation}
for $A, A'$ in the set of self-adjoint operators which are resolvent comparable with 
respect to $A_-$ (equivalently, $A_+$), that is, $A, A'$ satisfy 
for some (and hence for all) $\zeta \in \bbC \backslash \bbR$, 
\begin{equation}
\big[(A - \zeta I)^{-1} - (A_- - \zeta I)^{-1}\big], 
\big[(A' - \zeta I)^{-1} - (A_- - \zeta I)^{-1}\big]  \in \cB_1\big(L^2(\bbR)\big).  
\end{equation}
Thus, the hypotheses of \cite[Lemma~8.7.5]{Ya92} are satisfied and hence one 
obtains the following result: 

\begin{theorem} \lb{tB.8} 
Assume Hypothesis \ref{hB.0} and introduce the path $A_+(s)$, $s \in [0,1]$, as in 
\eqref{B.102}, with $A_+(0) = A_{+,1}$ $($cf.\ \eqref{A+n} with $n=1$$)$ and $A_+(1) = A_+$. 
Then for each $s \in [0,1]$, there exists a unique spectral shift function 
$\xi(\, \cdot \,; A_+(s), A_-)$ for the pair $(A_+(s), A_-)$ depending continuously on $s \in [0,1]$ in 
the space $L^1\big(\bbR; (\nu^2 +1)^{-1} d\nu\big)$, satisfying  
$\xi(\, \cdot \,; A_+(0), A_-) = \xi(\, \cdot \,; A_{+,1}, A_-)$, and $($cf.\ \eqref{B.36a}$)$, 
\begin{align} 
2 i \int_{\bbR} \f{\xi(\lambda; A_+(s),A_-) d \lambda}{\lambda^2 + 1}   
= {\tr}_{L^2(\bbR)} \big(\ln\big(U_+(s)U_-^{-1}\big)\big),     \lb{B.109}
\end{align} 
where
\begin{align}
U_- = (A_- - i I)(A_- + i I)^{-1}, \quad 
U_+(s) = (A_+(s) - i I)(A_+(s) + i I)^{-1}, \; s \in [0,1].
\end{align} 
In addition $($cf.\ \eqref{B.47a}$)$, 
\begin{equation}
\xi(\, \cdot \,; A_+(s),A_-) \in L^1(\bbR; d\nu), \quad s \in [0,1).   \lb{B.108} 
\end{equation}
\end{theorem} 

Thus, observing the equality $\chi_n(\cdot) = \wti \chi_{(1 - n^{-2})}(\cdot)$, Theorem \ref{tB.8} 
implies
\begin{align}
\begin{split}
\xi(\, \cdot \, ; A_+, A_-) &= \xi(\, \cdot \, ; A_+(1), A_-) 
= \lim_{s \uparrow 1}\xi(\, \cdot \, ; A_+(s), A_-) \\ 
&= \lim_{n \to \infty} \xi(\, \cdot \, ; A_{+,n}, A_-) \, 
\text{ in the norm $\|\cdot \|_{L^1(\bbR; (\nu^2 +1)^{-1} d\nu)}$, }
\end{split}
\end{align}
and therefore, a subsequence of $\{\xi(\, \cdot \, ; A_{+,n}, A_-)\}_{n \in \bbN}$ 
converges pointwise a.e. to $\xi(\, \cdot \, ; A_+, A_-)$ as $n \to \infty$. 
In particular, \eqref{B.96} shows that only $n_1 = 0$ in 
\eqref{B.22} is compatible with the family of spectral functions uniquely determined by 
Theorem \ref{tB.8}. Hence, on the basis of our approximation approach, one is naturally 
lead to the choice 
\begin{equation} 
\xi(\nu; A_+, A_-) = \f{1}{2 \pi} \int_{\bbR} dx\, \phi(x), \quad \nu \in \bbR, 
\end{equation}
which will henceforth be adopted for the remainder of this paper. 

We conclude this section with an elementary but useful consequence of Theorem \ref{tB.8}.

\begin{corollary} \lb{cB.9} 
Assume Hypothesis \ref{hB.0} and suppose that $f \in L^{\infty}(\bbR)$. Then
\begin{equation}
\lim_{n \to \infty} \|\xi(\, \cdot \, ; A_{+,n}, A_-) f 
- \xi(\, \cdot \, ; A_+, A_-) f\|_{L^1(\bbR; (\nu^2 + 1)^{-1}d\nu)} = 0,    \lb{B.118}
\end{equation}
in particular,
\begin{equation}
\lim_{n \to \infty} \int_{\bbR} \xi(\nu; A_{+,n}, A_-) d \nu \, g(\nu) 
= \int_{\bbR} \xi(\nu; A_+, A_-) d \nu \, g(\nu)     \lb{B.119}
\end{equation}
for all $g \in L^{\infty}(\bbR)$ such that 
$\esssup_{\nu \in \bbR} \big|(\nu^2 + 1) g(\nu)\big| < \infty$.
\end{corollary}
\begin{proof}
Relation \eqref{B.118} is clear from Theorem \ref{tB.8} and 
\begin{align}
\begin{split} 
& \|\xi(\, \cdot \, ; A_{+,n}, A_-) f 
- \xi(\, \cdot \, ; A_+, A_-) f\|_{L^1(\bbR; (\nu^2 + 1)^{-1}d\nu)}    \\
& \quad \leq \|f\|_{L^{\infty}(\bbR)} \, \|\xi(\, \cdot \, ; A_{+,n}, A_-) f 
- \xi(\, \cdot \, ; A_+, A_-) f\|_{L^1(\bbR; (\nu^2 + 1)^{-1}d\nu)},
\end{split} 
\end{align} 
and \eqref{B.119} is obvious from \eqref{B.118} and decomposing the (complex) measures 
\begin{equation} 
\xi(\nu; A_{+,n}, A_-) d \nu \, g(\nu) \, \text{ and } \, \xi(\nu; A_+, A_-) d \nu \, g(\nu) 
\end{equation} 
into 
\begin{equation} 
(\nu^2 + 1)^{-1} \xi(\nu; A_{+,n}, A_-) d \nu \, (\nu^2 + 1) g(\nu) \, \text{ and } \,  
(\nu^2 + 1)^{-1} \xi(\nu; A_+, A_-) d \nu \, (\nu^2 + 1) g(\nu).
\end{equation}   
\end{proof}

\section{The $(1+1)$-Dimensional Example} \lb{s4}

In this section we start the discussion of an interesting example that does 
not satisfy the relative trace class condition Hypothesis~2.1\,$(iv)$ in 
\cite{CGPST14a} and \cite{GLMST11}. 

To this end, we now strengthen Hypothesis \ref{hB.0} as follows: 

\begin{hypothesis} \lb{h3.1} 
Suppose the real-valued functions $\phi, \theta$ satisfy  
\begin{align}
& \phi \in AC_{\loc}(\bbR) \cap L^{\infty}(\bbR) \cap L^1(\bbR),  
\; \phi' \in L^{\infty}(\bbR),     \\
\begin{split} 
& \theta \in AC_{\loc}(\bbR) \cap L^{\infty}(\bbR), \; 
\theta' \in L^{\infty}(\bbR) \cap L^1(\bbR),      \\
& \lim_{t \to \infty} \theta (t) = 1, \; \lim_{t \to - \infty} \theta (t) = 0. 
\end{split} 
\end{align} 
\end{hypothesis}

Given Hypothesis \ref{h3.1}, we now introduce the family of self-adjoint operators 
$A(t)$, $t \in \bbR$, in $L^2(\bbR)$, 
\begin{equation}
A(t) = - i \f{d}{dx} + \theta(t) \phi, \quad 
\dom(A(t)) = W^{1,2}(\bbR), \quad t \in \bbR.
\end{equation}
(In fact, given $\theta \in L^\infty(\bbR)$, self-adjointness of $A(t)$, is equivalent 
to the condition $\phi \in L^2_{\locunif}(\bbR)$, see the references in \cite{GW13}). Its  
asymptotes, $A_- = - i \f{d}{dx}$, $A_+ = - i \f{d}{dx} + \phi$, 
$\dom(A_{\pm}) = W^{1,2}(\bbR)$, were studied in detail in Section \ref{s3}.    
Then a simple application of resolvent identities proves that 
\begin{equation}
\nlim_{t \to \pm \infty} (A(t) - z I)^{-1} = (A_{\pm} - z I)^{-1}, \quad 
z \in \bbC \backslash \bbR. 
\end{equation}
Indeed, it suffices to note that for $t\in \bbR$, $z \in \bbC \backslash \bbR$, 
\begin{align}
& \big\|(A(t) - zI)^{-1} - (A_- - zI)^{-1}\big\|_{\cB(L^2(\bbR))} 
\leq |\Im(z)|^{-2} \|\phi\|_{L^{\infty}(\bbR)} |\theta(t)|,    \\
& \big\|(A(t) - zI)^{-1} - (A_+ - zI)^{-1}\big\|_{\cB(L^2(\bbR))} 
\leq |\Im(z)|^{-2} \|\phi\|_{L^{\infty}(\bbR)} |\theta(t) - 1|,
\end{align}
employing 
$A(t) = A_- + \theta(t) \phi = A_+ + [\theta(t) - 1] \phi$, $t \in \bbR$. 
Moreover, as in \eqref{3.13} one also obtains, 
\begin{equation}
\big[(A(t) - z I)^{-1} - (A_- - z I)^{-1}\big] \in \cB_1\big(L^2(\bbR)\big), \quad 
t \in \bbR, \; z \in \bbC \backslash \bbR. 
\end{equation}
As in \eqref{3.5} and \eqref{3.5A}, introducing the unitary operator of multiplication 
$U(t) = e^{- i \theta(t) \int_0^x dx' \phi(x')}$, $t \in \bbR$, in $L^2(\bbR)$, one obtains 
$A(t) = U(t) A_- U(t)^{-1}$, $t \in \bbR$. 

It will be convenient to introduce the family of bounded operators $B(t)$, $t \in \bbR$, 
in $L^2(\bbR)$, where
\begin{equation}
B(t) = \theta(t) \phi, \quad \dom(B(t)) = L^{2}(\bbR), \;  t \in \bbR,
\end{equation}
implying $A(t) = A_- + B(t)$, $t \in \bbR$.

Next, we introduce the operator $d/dt$ in $L^2\big(\bbR; dt; L^2(\bbR;dx)\big)$  by 
\begin{align}
& \bigg(\f{d}{dt}f\bigg)(t) = f'(t) \, \text{ for a.e.\ $t\in\bbR$,}    \no \\
& \, f \in \dom(d/dt) = \big\{g \in L^2\big(\bbR;dt;L^2(\bbR)\big) \, \big|\,
g \in AC_{\loc}\big(\bbR; L^2(\bbR)\big), \\
& \hspace*{6.3cm} g' \in L^2\big(\bbR;dt;L^2(\bbR)\big)\big\}   \no \\
& \hspace*{2.3cm} = W^{1,2} \big(\bbR; dt; L^2(\bbR; dx)\big).      \label{2.ddt} 
\end{align} 
At this point we turn to the pair $(\bsH_2, \bsH_1)$ and identify 
$L^2\big(\bbR; dt; L^2(\bbR; dx)\big) = L^2(\bbR^2; dt dx)$ from now on, and for simplicity, typically 
abbreviate the latter by $L^2(\bbR^2)$. We start by introducig the model operator 
$\bsD_\bsA^{}$ in $L^2(\bbR^2)$ by 
\begin{equation}
\bsD_\bsA^{} = \f{d}{dt} + \bsA,
\quad \dom(\bsD_\bsA^{})= W^{1,2}(\bbR^2).    \lb{2.DA}
\end{equation}
Clearly, $\bsD_\bsA^{}$ is densely defined and closed (cf.\ \cite[Lemma~4.4]{GLMST11}). 
Similarly, the adjoint operator $\bsD_\bsA^*$ of $\bsD_\bsA^{}$ in 
$L^2(\bbR^2)$ is then given by
\begin{equation}
\bsD_\bsA^*=- \f{d}{dt} + \bsA, \quad
\dom(\bsD_\bsA^*) = W^{1,2}(\bbR^2).   
\end{equation}
Following a tradition in mathematical physics, we dubbed the model represented 
by $\bsD_\bsA^{}$ a $(1+1)$-dimensional model due 
to the underlying one-dimensionality of $x \in \bbR$ and $t \in \bbR$. 

In addition, we introduce $\bsA_-$  in $L^2(\bbR^2)$, the self-adjoint (constant fiber) 
operator defined by 
\begin{align}
& (\bsA_- f)(t) = A_- f(t) \, \text{ for a.e.\ $t\in\bbR$,}   \no \\
& f \in \dom(\bsA_-) = \bigg\{g \in L^2(\bbR^2) \,\bigg|\,
g(t, \cdot)\in \dom(A_-) \text{ for a.e.\ } t\in\bbR,    \no \\
& \quad t \mapsto A_- g(t, \cdot) \text{ is (weakly) measurable,} \,  
\int_{\bbR} dt \, \|A_- g(t)\|_{L^2(\bbR)}^2 < \infty\bigg\}.    \lb{2.DA-}
\end{align} 

Then the operators $\bsH_j$, $j=1,2$, are defined by (see our discussion in 
\eqref{2.7a}--\eqref{2.9a}, for convenience, we repeat this at this point),  
\begin{align}
& \bsH_1 = \bsD_{\bsA}^{*} \bsD_{\bsA}^{} = - \f{\partial^2}{\partial t^2} - \f{\partial^2}{\partial x^2} 
- 2 i \theta(t) \phi(x) \f{\partial}{\partial x}   \no \\ 
& \hspace*{2.6cm} - \theta'(t) \phi(x) - i \theta(t) \phi'(x) + \theta^2(t) \phi(x)^2,    \lb{3.H1} \\
& \bsH_2 = \bsD_{\bsA}^{} \bsD_{\bsA}^{*} = - \f{\partial^2}{\partial t^2} - \f{\partial^2}{\partial x^2} 
- 2 i \theta(t) \phi(x) \f{\partial}{\partial x}   \no \\ 
& \hspace*{2.6cm} + \theta'(t) \phi(x) - i \theta(t) \phi'(x) + \theta^2(t) \phi(x)^2,     \lb{3.H2} \\
& \dom(\bsH_1) =  \dom(\bsH_2) = W^{2,2}\big(\bbR^2\big).     \lb{domHj}
\end{align} 

Thus,
\begin{align}
\begin{split} 
(\bsH_2 - z \, \bsI)^{-1} - (\bsH_1 - z \, \bsI)^{-1} = 
- (\bsH_1 - z \, \bsI)^{-1} [2 \theta' \phi] (\bsH_2 - z \, \bsI)^{-1},&   \lb{3.31} \\
z\in\rho(\bsH_1) \cap \rho(\bsH_2).& 
\end{split} 
\end{align}
Since by hypothesis, $\theta' \phi \in L^1(\bbR^2; dt dx)$, one can once more 
apply \cite[Theorem\ 4.1]{Si05} and conclude that 
\begin{equation}
\big[(\bsH_2 - z \, \bsI)^{-1} - (\bsH_1 - z \, \bsI)^{-1}\big] \in 
\cB_1\big(L^2(\bbR^2)\big), \quad z \in \rho(\bsH_1) \cap \rho(\bsH_2), 
\lb{3.32}
\end{equation} 
again by decomposing $[\theta' \phi] = |\theta'|^{1/2} |\phi|^{1/2} \sgn(\theta') 
 \sgn(\phi) |\theta'|^{1/2} |\phi|^{1/2}$ and using that
\begin{align}
& \ol{(\bsH_1 - z \, \bsI)^{-1} (\bsH_0 - z \, \bsI)} 
= \big[(\bsH_0 - {\ol z} \, \bsI) (\bsH_1 - {\ol z} \, \bsI)^{-1}\big]^*  
\in \cB\big(L^2(\bbR^2)\big),    \lb{3.33} \\
& (\bsH_0 - z \, \bsI) (\bsH_2 - z \, \bsI)^{-1} \in \cB\big(L^2(\bbR^2)\big), 
\quad z \in \rho(\bsH_1) \cap \rho(\bsH_2),      \lb{3.34}
\end{align} 
and hence 
\begin{align}
& (\bsH_2 - z \, \bsI)^{-1} - (\bsH_1 - z \, \bsI)^{-1} 
= - \big[\ol{(\bsH_1 - z \, \bsI)^{-1}(\bsH_0 - z \, \bsI)}\big]    \no \\
& \quad \times \big[(\bsH_0 -z \, \bsI)^{-1} [2 \theta' \phi] (\bsH_0 -z \, \bsI)^{-1} \big] 
\big[(\bsH_0 -z \, \bsI)(\bsH_2  - z \, \bsI)^{-1}\big],     \lb{3.35} \\
& \hspace*{8.1cm} z \in \bbC \backslash [0,\infty),    \no  
\end{align} 
where $\bsH_0$ in $L^2(\bbR^2)$ abbreviates 
\begin{equation}
\bsH_0 = - \f{d^2}{dt^2} + \bsA_-^2 = \bigg(- \f{\partial^2}{\partial t^2} - \f{\partial^2}{\partial x^2} \bigg) 
= - \Delta, \quad \dom(\bsH_0) = W^{2,2}(\bbR^2).  
\lb{H0}
\end{equation}
 
The fact \eqref{3.32} implies that the spectral shift function 
$\xi(\, \cdot \, ; \bsH_2, \bsH_1)$ for the pair $(\bsH_2, \bsH_1)$ is well-defined, satisfies
\begin{equation}
\xi(\, \cdot \, ; \bsH_2, \bsH_1) \in L^1\big(\bbR; (\lambda^2 + 1)^{-1} d\lambda\big), 
\end{equation} 
and since $\bsH_j\geq 0$, $j=1,2$, one uniquely introduces $\xi(\,\cdot\,; \bsH_2,\bsH_1)$ 
by requiring that
\begin{equation}
\xi(\lambda; \bsH_2,\bsH_1) = 0, \quad \lambda < 0,    \lb{2.46c}
\end{equation}
implying the Krein--Lifshits trace formula,  
\begin{align}
\begin{split}
\tr_{L^2(\bbR^2;dtdx)} \big((\bsH_2 - z \, \bsI)^{-1} - (\bsH_1 - z \, 
\bsI)^{-1}\big)
= - \int_{[0, \infty)}  \frac{\xi(\lambda; \bsH_2, \bsH_1)  
d\lambda}{(\lambda -z)^2},&     \\
z\in\bbC\backslash [0,\infty).&    \lb{3.45}
\end{split} 
\end{align} 

Introducing $\bsB$ and $\bsB^{\prime}$ in terms of the bounded operator families 
$B(t)$, $B'(t)$, $t \in \bbR$, in analogy to \eqref{1.S}, one can decompose 
$\bsH_j$, $j=1,2$, as follows:
\begin{align}
\bsH_j &= \f{d^2}{dt^2} + \bsA^2 + (-1)^j \bsA^{\prime}    \\
&= \bsH_0 + \bsB \bsA_- + \bsA_- \bsB + \bsB^2 + (-1)^j \bsB^{\prime}, \quad j =1,2. 
\end{align}
One notes that the operators $\bsH_j$, $j=1,2,$ are well-defined, 
since $\bsB$ leaves the domain of $\bsA_-$ invariant. In addition, since 
one can write 
\begin{equation} 
\bsB\bsA_-+\bsA_-\bsB = [\bsA_-,\bsB]+2\bsB\bsA_-
= - [\bsA_-,\bsB]+2\bsA_-\bsB, 
\end{equation} 
and 
\begin{align}
([\bsA_-,\bsB]f)(t)&=(\bsA_-\bsB f)(t)-(\bsB\bsA_- f)(t)=A_-\theta(t)\phi f(t)-\theta(t)\phi A_-f(t)   \no \\
&=\theta(t)[A_-,\phi ]f(t)= - i \theta(t)\phi' f(t), 
\quad f\in W^{2,2}(\mathbb{R}^2),  
\end{align}
employing the fact that $\phi', \theta \in L^\infty(\bbR)$, one obtains 
that the commutator $[\bsA_-,\bsB]$ has a bounded closure.
For subsequent purposes we denote 
\begin{equation}\label{commutBA}
\bsC:= \ol{[\bsA_-,\bsB]}
\end{equation}
and write
\begin{align}\label{HwithC}
\begin{split} 
\bsH_j &= \bsH_0 + 2\bsB\bsA_- +\bsC+ \bsB^2 + (-1)^j \bsB^{\prime}   \\
&= \bsH_0 + 2\bsA_- \bsB-\bsC+ \bsB^2 + (-1)^j \bsB^{\prime}, \quad j =1,2.
\end{split} 
\end{align} 

Returning to the approximations introduced in \eqref{def_chi_n}--\eqref{2.58}, we now 
introduce 
\begin{equation}
A_n(t) = A_- + \theta(t) \chi_n(A_-) \phi \chi_n(A_-), \quad 
\dom(A_n(t)) = \dom(A_-), \quad n \in \bbN, \; t \in \bbR,   \lb{dAn}  
\end{equation}
and note
\begin{align} 
& A_n'(t) = B_n'(t) = \theta'(t) \chi_n(A_-) \phi \chi_n(A_-) 
\in \cB_1\big(L^2(\bbR)\big), \quad n \in \bbN, \; t \in \bbR,  \\
& \int_{\bbR} dt \, \|A_n'(t)\|_{\cB_1(L^2(\bbR))} 
\leq \|\chi_n(A_-) \phi \chi_n(A_-)\|_{\cB_1(L^2(\bbR))} \, 
\|\theta'\|_{L^1(\bbR)} < \infty, \quad n \in \bbN,      \lb{2.57} 
\end{align} 
recalling \eqref{2.61}. In addition, introducing the decompositions,
\begin{align}
& \bsH_{j,n} = \f{d^2}{dt^2} + \bsA_n^2 + (-1)^j \bsA_n^{\prime}    \no \\
& \hspace*{8mm} = \bsH_0 + \bsB_n \bsA_- + \bsA_- \bsB_n 
+ \bsB_n^2 + (-1)^j \bsB_n^{\prime}, \\
& \dom(\bsH_{j,n}) = \dom(\bsH_0) = W^{2,2}(\bbR^2),  
\quad n \in \bbN, \;  j =1,2,  \no 
\end{align}
with
\begin{align}
& \bsB_n = \chi_n(\bsA_-) \bsB \chi_n(\bsA_-), \quad 
\bsB_n^{\prime} = \chi_n(\bsA_-) \bsB^{\prime} \chi_n(\bsA_-), \quad n \in \bbN, \\
& \bsC_n=\chi_n(\bsA_-) \bsC \chi_n(\bsA_-), \quad n \in \bbN,     \lb{commutBAn}
\end{align} 
one can write 
\begin{align} 
\begin{split} 
\bsH_{j,n} &= \bsH_0 + 2\bsB_n \bsA_- +\bsC_n + \bsB_n^2 + (-1)^j \bsB_n^{\prime}
\\
&= \bsH_0 + 2\bsA_-\bsB_n - \bsC_n + \bsB_n^2 + (-1)^j \bsB_n^{\prime} 
\quad n \in \bbN, \; j =1,2.  
\end{split} 
\end{align} 

\begin{lemma} \lb{l3.5}
Assume Hypothesis \ref{h3.1} and let $z \in \bbC \backslash [0,\infty)$. Then the following assertions hold: \\
$(i)$ The operators $\bsH_{j,n}$ converge to $\bsH_j$, $j=1,2$, 
in the strong resolvent sense, 
\begin{equation}
\slim_{n \to \infty} (\bsH_{j,n}-z\, \bsI)^{-1} 
= ( \bsH_{j}-z\, \bsI)^{-1}, \quad j=1,2.     \lb{limR}
\end{equation}
$(ii)$ The operators 
\begin{equation}
\ol{( \bsH_{j,n}-z \, \bsI)^{-1}(\bsH_{0}-z \, \bsI)} = 
\big[(\bsH_{0}- {\ol z} \, \bsI) ( \bsH_{j,n}- {\ol z} \, \bsI)^{-1}\big]^*, \quad  j=1,2, \; 
n \in \bbN,     \lb{closure} 
\end{equation}
and 
$(\bsH_{0}- z \, \bsI) ( \bsH_{j,n}- z \, \bsI)^{-1}$, $j=1,2$, $n \in \bbN$, 
are uniformly bounded with respect to $n \in \bbN$, that is, there exists 
$C \in (0,\infty)$ such that
\begin{equation}
\big\|(\bsH_{0}- z \, \bsI) ( \bsH_{j,n}- z \, \bsI)^{-1}\big\|_{\cB(L^2(\bbR^2))} 
\leq C, \quad j=1,2, \; n \in \bbN.    \lb{bd} 
\end{equation}
In addition, 
\begin{align} 
& \slim_{n \to \infty} \ol{(\bsH_{j,n}-z \, \bsI)^{-1}(\bsH_{0} - z \, \bsI)} 
= \ol{( \bsH_{j}-z \, \bsI)^{-1}(\bsH_{0}-z \, \bsI)}, \quad j=1,2,   \lb{conv1} \\
& \slim_{n \to \infty} (\bsH_{0}-z \, \bsI) (\bsH_{j,n}-z\, \bsI)^{-1} 
= (\bsH_{0}-z \, \bsI)( \bsH_{j}-z\, \bsI)^{-1}, \quad j=1,2.    \lb{conv2}
\end{align}  
\end{lemma}
\begin{proof} Since the proof for the operators $\bsH_{2,n},\bsH_2$ is a verbatim repetition of the proof for $\bsH_{1,n},\bsH_1$, we exclusively focus on the latter. 

\noindent 
$(i)$ By \eqref{HwithC} and the analogous equation for $\bsH_{1,n}$, 
the operators $\bsH_1$ and $\bsH_{1,n}$ have a common core 
$\mathrm{dom}(\bsH_1)$. Since 
\begin{equation} 
\bsB^\prime-\bsB^\prime_n=\bsB^\prime-\chi_n(\bsA_-) \bsB^\prime \chi_n(\bsA_-)=(\bsI-\chi_n(\bsA_-))\bsB^\prime+\chi_n(\bsA_-)\bsB^\prime(\bsI-\chi_n(\bsA_-)), 
\end{equation} 
and $\bsB^\prime$ is a bounded operator, the convergence 
\begin{equation}
\slim_{n \to \infty} \bsB^\prime_n = \bsB^\prime    \lb{limBprime}
\end{equation}
holds. Arguing analogously, one also obtains that 
\begin{equation} 
\slim_{n \to \infty} \bsB_n = \bsB, \quad 
\slim_{n \to \infty} \bsC_n = \bsC.     \lb{limBC}
\end{equation} 
Next, rewriting
\begin{align}
& \bsB^2-\bsB^2_n=\bsB^2-\chi_n(\bsA_-)\bsB\chi_n(\bsA_-)\bsB\chi_n(\bsA_-) 
 \\
& \quad =(\bsI-\chi_n(\bsA_-))\bsB^2+\chi_n(\bsA_-)\bsB\big(\bsB(\bsI-\chi_n(\bsA_-))+(\bsI-\chi_n(\bsA_-))\bsB\chi_n(\bsA_-)\big),   \no 
\end{align}
one obtains also that
\begin{equation} 
\slim_{n \to \infty} \bsB^2_n = \bsB^2.    \lb{limB-2}
\end{equation} 
Thus, it remains to show that $\slim_{n \to \infty} \bsB_n\bsA_- f = \bsB\bsA_-f$ for all $f\in\mathrm{dom}(\bsH_1)$. Indeed, one verifies  
\begin{align}
\bsB\bsA_--\bsB_n\bsA_-&=\bsB\bsA_--\chi_n(\bsA_-)\bsB\bsA_-\chi_n(\bsA_-)\\
&=(\bsI-\chi_n(\bsA_-))\bsB\bsA_-+\chi_n(\bsA_-)\bsB(\bsI-\chi_n(\bsA_-))\bsA_-,  
\no 
\end{align} 
implying the required convergence. Consequently, 
\begin{equation} 
\slim_{n \to \infty} \bsH_{1,n}f = \bsH_1f, \quad f\in \mathrm{dom}(\bsH_1).  
\end{equation} 
Since $\bsH_{1,n}$ and $\bsH_1$ are self-adjoint operators with a common core, \cite[Theorem~VIII.25]{RS80} (see also \cite[Theorem~9.16]{We80}) implies 
that $\bsH_{1,n}$ converges to $\bsH_1$ in the strong resolvent sense.

\smallskip 
\noindent 
$(ii)$ Fix $z \in \bbC \backslash [0,\infty)$. First, one observes that 
\begin{equation}
\ol{( \bsH_{1,n}-z \, \bsI)^{-1}(\bsH_{0}-z \, \bsI)} = 
\big[(\bsH_{0}- {\ol z} \, \bsI) ( \bsH_{1,n}- {\ol z} \, \bsI)^{-1}\big]^*. 
\end{equation}
Using the standard resolvent identity, one obtains
\begin{equation}\label{ssss}
({\bsH}_{1,n}-z \, \bsI)^{-1}-(\bsH_0-z \, \bsI)^{-1}= 
-({\bsH}_{1,n}-z \, \bsI)^{-1}\big[(\bsH_1-\bsH_0)(\bsH_0-z \, \bsI)^{-1}\big], 
\end{equation}
and hence concludes,
\begin{align}
\begin{split} 
& (\bsH_0 - z \, \bsI)(\bsH_{1,n} - z \, \bsI)^{-1} 
= \bsI - \big[(\bsH_{1,n} - \bsH_0)(\bsH_{1,n} - z \, \bsI)^{-1}\big]   \\
& \quad = \bsI - \big[2 \bsB_n \bsA_- + \bsC_n + \bsB_n^2 - \bsB_n^{\prime}\big]
(\bsH_{1,n} - z \, \bsI)^{-1}, \quad n \in \bbN.   \lb{seq} 
\end{split}
\end{align}
Since strongly convergent sequences of bounded operators are uniformly 
bounded in norm, \eqref{limR}, \eqref{limBprime}, \eqref{limBC}, and 
\eqref{limB-2}  yield uniform boundedness of \eqref{seq} with respect to 
$n \in \bbN$, except for the term $-2\bsB_n \bsA_- (\bsH_{1,n} - z \, \bsI)^{-1}$, 
which we focus on next. One obtains 
\begin{align}
& \bsA_- (\bsH_{1,n} - z \, \bsI)^{-1} = \bsA_- (\bsH_0 - z \, \bsI)^{-1}   \no \\ 
& \qquad - \bsA_- (\bsH_0 - z \, \bsI)^{-1}\big[(\bsH_{1,n} - \bsH_0)
(\bsH_{1,n} - z \, \bsI)^{-1}\big]   \no \\
& \quad = \bsA_- (\bsH_0 - z \, \bsI)^{-1}    \lb{seq1} \\ 
& \qquad - \bsA_- (\bsH_0 - z \, \bsI)^{-1}
\big[2 \bsA_- \bsB_n - \bsC_n + \bsB_n^2 - \bsB_n^{\prime}\big]
(\bsH_{1,n} - z \, \bsI)^{-1}, \quad n \in \bbN,   \no 
\end{align}
which is uniformly bounded with respect to $n \in \bbN$ since
\begin{equation}
\ol{\bsA_- (\bsH_0 + \bsI)^{-1} \bsA_-} = 
\big[\bsA_- (\bsH_0 + \bsI)^{-1/2}\big] \big[\bsA_- (\bsH_0 + \bsI)^{-1/2}\big]^* 
\in \cB\big(L^2(\bbR^2)\big). 
\end{equation}
This proves \eqref{bd}. 
In fact, gathering all terms from \eqref{seq} and \eqref{seq1} results in 
\begin{align}
& (\bsH_0 - z \, \bsI)(\bsH_{1,n} - z \, \bsI)^{-1} 
= \bsI - \big[2 \bsB_n \bsA_- + \bsC_n + \bsB_n^2 - \bsB_n^{\prime}\big]
(\bsH_{1,n} - z \, \bsI)^{-1}    \no \\
& \quad = \bsI - \big[\bsC_n + \bsB_n^2 - \bsB_n^{\prime}\big]
(\bsH_{1,n} - z \, \bsI)^{-1} - 2 \bsB_n \bsA_- (\bsH_0 - z \, \bsI)^{-1}    \no \\
& \qquad + 2 \bsB_n \bsA_- (\bsH_0 - z \, \bsI)^{-1}  
\big[2 \bsA_- \bsB_n - \bsC_n + \bsB_n^2 - \bsB_n^{\prime}\big] 
(\bsH_{1,n} - z \, \bsI)^{-1}   \no \\
& \quad = \bsI - \big[\bsC_n + \bsB_n^2 - \bsB_n^{\prime}\big]
(\bsH_{1,n} - z \, \bsI)^{-1} - 2 \bsB_n \big[\bsA_- (\bsH_0 - z \, \bsI)^{-1}\big]   \no \\
& \qquad + 2 \bsB_n \big[\bsA_- (\bsH_0 - z \, \bsI)^{-1}\big]   
\big[- \bsC_n + \bsB_n^2 - \bsB_n^{\prime}\big] 
(\bsH_{1,n} - z \, \bsI)^{-1}    \no \\
& \qquad + 4 \bsB_n \big[\ol{\bsA_- (\bsH_0 - z \, \bsI)^{-1} \bsA_-}\big] \bsB_n 
(\bsH_{1,n} - z \, \bsI)^{-1}, \quad n \in \bbN.    \lb{seq3}
\end{align}
Since a finite number of products of strongly convergent sequences of (necessarily uniformly) bounded 
operators is strongly convergent, \eqref{seq3} proves the strong convergence in 
\eqref{conv2}.   

Finally, to prove \eqref{conv1} it suffices to combine the strong resolvent convergence in \eqref{limR}, the uniform boundedness in \eqref{bd} with equality \eqref{closure}, the strong convergence   
\begin{equation} 
\slim_{n \to \infty} (\bsH_{1,n}-z \, \bsI)^{-1}(\bsH_{0} - z \, \bsI) f 
= ( \bsH_1 - z \, \bsI)^{-1}(\bsH_{0}-z \, \bsI)f, \quad f \in \dom(\bsH_0), 
\end{equation} 
and the fact that $\dom(\bsH_0)$ is dense in $L^2(\bbR^2)$. Here we 
used that  uniformly bounded sequences of bounded operators in a Hilbert space converge strongly if they converge pointwise on a dense subset of the Hilbert 
space. 
\end{proof}

Next, we recall that 
\begin{align}
 \bsH_2 - \bsH_1 = 2 \bsB^{\prime},  \quad 
 \bsH_{2,n} - \bsH_{1,n} = 2 \bsB_n^{\prime} = 2 \chi_n(\bsA_-) \bsB^{\prime} \chi_n(\bsA_-), 
 \; n \in \bbN,    \lb{2.65} 
\end{align}
and in analogy to \eqref{3.31}--\eqref{3.34} one concludes that 
\begin{equation}
\big[(\bsH_{2,n}-z \, \bsI)^{-1} - (\bsH_{1,n}-z \, \bsI)^{-1}\big] \in \cB_1\big(L^2(\bbR^2)\big), 
\quad n \in \bbN, \; z \in \bbC \backslash [0,\infty),    \lb{2.65a}
\end{equation}
since 
\begin{align}
& (\bsH_{2,n}-z \, \bsI)^{-1} - (\bsH_{1,n}-z \, \bsI)^{-1} 
= - \big[\ol{(\bsH_{1,n}-z \, \bsI)^{-1}(\bsH_0 - z \, \bsI)}\big] \chi_n(\bsA_-)   \no \\
& \quad \times \big[(\bsH_0 -z \, \bsI)^{-1} 2 \bsB^{\prime} (\bsH_0 -z \, \bsI)^{-1} \big] 
\chi_n(\bsA_-)  \big[(\bsH_0 -z \, \bsI)(\bsH_{2,n} - z \, \bsI)^{-1}\big],     \lb{2.65b} \\
& \hspace*{8.2cm} n \in \bbN, \; z \in \bbC \backslash [0,\infty),    \no  
\end{align}
employing commutativity of $\chi_n(\bsA_-)$ and $(\bsH_0 -z \, \bsI)^{-1}$, that is, 
\begin{equation}
\big[\chi_n(\bsA_-), (\bsH_0 -z \, \bsI)^{-1}\big] = 0.    \lb{comm}
\end{equation}

Finally, we proceed to some crucial convergence results to be used in Section \ref{s5}.

\begin{theorem} \lb{t3.7}
Assume Hypothesis \ref{h3.1} and let $z \in \bbC \backslash [0,\infty)$. Then 
\begin{align}
\begin{split} 
& \lim_{n\to\infty} \big\|\big[(\bsH_{2,n} - z \, \bsI)^{-1} - (\bsH_{1,n} - z \, \bsI)^{-1}\big]  \\
& \hspace*{1cm} - [(\bsH_2 - z \, \bsI)^{-1} - (\bsH_1 - z \, \bsI)^{-1}\big]
\big\|_{\cB_1(L^2(\bbR^2))} = 0.    \lb{2.66}
\end{split} 
\end{align}
\end{theorem}
\begin{proof}
An application of \eqref{2.65} and the resolvent equation for the difference of 
resolvents in \eqref{2.66} yield  
\begin{align}
& \big[(\bsH_{2,n} - z \, \bsI)^{-1} - (\bsH_{1,n} - z \, \bsI)^{-1}\big]
 - [(\bsH_2 - z \, \bsI)^{-1} - (\bsH_1 - z \, \bsI)^{-1}\big]    \no \\
 & \quad = - 2 (\bsH_{2,n} - z \, \bsI)^{-1} \bsB_n^{\prime} (\bsH_{1,n} - z \, \bsI)^{-1} 
 + 2 (\bsH_2 - z \, \bsI)^{-1} \bsB^{\prime} (\bsH_1 - z \, \bsI)^{-1}   \no \\
 & \quad = - 2\ol{\big[(\bsH_{2,n} - z \, \bsI)^{-1} (\bsH_0 - z \, \bsI)\big]}    \no \\
& \qquad \quad \times  \big\{\chi_n(\bsA_-) (\bsH_0 - z \, \bsI)^{-1} 
\bsB^{\prime} (\bsH_0 - z \, \bsI)^{-1} 
 \chi_n(\bsA_-)\big\}    \no \\
& \qquad \quad \times \big[(\bsH_0 - z \, \bsI) (\bsH_{1,n} - z \, \bsI)^{-1}\big]      \\
& \qquad + 2\ol{\big[(\bsH_2 - z \, \bsI)^{-1} (\bsH_0 - z \, \bsI)\big]}    \no \\
& \qquad \quad \times  \big\{(\bsH_0 - z \, \bsI)^{-1} \bsB^{\prime} 
(\bsH_0 - z \, \bsI)^{-1}\big\}    \no \\
& \qquad \quad \times \big[(\bsH_0 - z \, \bsI) (\bsH_1 - z \, \bsI)^{-1}\big], \quad 
z \in \bbC \backslash [0,\infty).  
\end{align}
By Lemma \ref{l3.6}, the term $\big\{\chi_n(\bsA_-) (\bsH_0 - z \, \bsI)^{-1} 
\bsB^{\prime} (\bsH_0 - z \, \bsI)^{-1} 
 \chi_n(\bsA_-)\big\}$ converges to $\big\{(\bsH_0 - z \, \bsI)^{-1} \bsB^{\prime} 
(\bsH_0 - z \, \bsI)^{-1}\big\}$ in $\cB_1\big(L^2(\bbR^2)\big)$-norm as $n\to \infty$.
Another application of Lemma \ref{l3.6} proves \eqref{2.66} since by Lemma \ref{l3.5}\,$(ii)$, for $z \in \bbC \backslash [0,\infty)$, one has  
\begin{align}
& \slim_{n\to\infty} \big[(\bsH_0 - z \, \bsI) (\bsH_{1,n} - z \, \bsI)^{-1}\big] = 
\big[(\bsH_0 - z \, \bsI) (\bsH_1 - z \, \bsI)^{-1}\big],    \lb{2.69} \\
& \slim_{n\to\infty} \ol{\big[(\bsH_{2,n} - z \, \bsI)^{-1} (\bsH_0 - z \, \bsI)\big]} =  
\ol{\big[(\bsH_2 - z \, \bsI)^{-1} (\bsH_0 - z \, \bsI)\big]}.     \lb{2.70}
\end{align} 
\end{proof}

\begin{theorem} \lb{t3.8}
Assume Hypothesis \ref{h3.1} and let $z, z' \in \bbC \backslash [0,\infty)$. Then 
\begin{align}
& \lim_{n\to\infty} \big\|\bsB_n^{\prime} (\bsH_{j,n} - z \, \bsI)^{-1} 
- \bsB^{\prime} (\bsH_j - z \, \bsI)^{-1}
\big\|_{\cB_2(L^2(\bbR^2))} = 0, \quad j = 1,2,   \lb{2.71} \\
& \lim_{n\to\infty} \big\|(\bsH_{j,n} - z \, \bsI)^{-1} 2 \bsB_n^{\prime} (\bsH_{j,n} - z' \, \bsI)^{-1} 
\no \\
& \hspace*{1cm} - (\bsH_j - z \, \bsI)^{-1} 2 \bsB^{\prime} (\bsH_j - z' \, \bsI)^{-1}
\big\|_{\cB_1(L^2(\bbR^2))} = 0,  \quad j = 1,2.    \lb{2.72}
\end{align}
\end{theorem}
\begin{proof}
To prove \eqref{2.71} one writes
\begin{align}
\begin{split} 
& \bsB_n^{\prime} (\bsH_{j,n} - z \, \bsI)^{-1} = \chi_n(\bsA_-) \big[\bsB^{\prime}  
 (\bsH_0 - z \, \bsI)^{-1}\big]    \\
& \quad \times \chi_n(\bsA_-) \big[(\bsH_0 - z \, \bsI) (\bsH_{j,n} - z \, \bsI)^{-1}\big], 
\quad j =1,2, \; n \in \bbN, 
\end{split} 
\end{align}
employing once again commutativity of $\chi_n(\bsA_-)$ and $(\bsH_0 - z \, \bsI)^{-1}$ 
(cf.\ \eqref{comm}). Since 
\begin{equation} 
\bsB^{\prime} (\bsH_0 - z \, \bsI)^{-1} \in \cB_2\big(L^2(\bbR^2)\big), \quad 
z \in \bbC \backslash [0,\infty),    \lb{2.74}
\end{equation} 
by \cite[Theorem~4.1]{Si05}, \eqref{2.71} is a consequence of Lemma \ref{l3.6} combined with 
\begin{equation} 
\slim_{n \to \infty} \chi_n(\bsA_-) = \bsI    \lb{2.75} 
\end{equation}
(applying the spectral theorem, see also \eqref{slim}) and \eqref{conv2}. 

Relation \eqref{2.72} follows along exactly the same lines upon decomposing 
\begin{align}
& (\bsH_{j,n} - z \, \bsI)^{-1} 2 \bsB_n^{\prime} (\bsH_{j,n} - z' \, \bsI)^{-1} = 
\big[\ol{(\bsH_{j,n} - z \, \bsI)^{-1} (\bsH_0 - z \, \bsI)}\big] \chi_n(\bsA_-)      \no \\
& \quad \times \big[(\bsH_0 - z \, \bsI)^{-1} 2 \bsB^{\prime} (\bsH_0 - z' \, \bsI)^{-1}\big] 
\chi_n(\bsA_-) \big[(\bsH_0 - z' \, \bsI) (\bsH_{j,n} - z' \, \bsI)^{-1}\big],     \no \\
& \hspace*{8.6cm}  j =1,2, \; n \in \bbN,     
\end{align} 
applying once more Lemma \ref{l3.6}, \eqref{conv1}, \eqref{conv2}, \eqref{comm}, \eqref{2.75}, 
and 
\begin{align} 
\begin{split}  
& (\bsH_0 - z \, \bsI)^{-1} \bsB^{\prime} (\bsH_0 - z' \, \bsI)^{-1} 
= \big[(\bsH_0 - z \, \bsI)^{-1} |\bsB^{\prime}|^{1/2}\big] \sgn(\bsB^{\prime})    \\
& \quad \times \big[ |\bsB^{\prime}|^{1/2} (\bsH_0 - z' \, \bsI)^{-1}\big] 
\in \cB_1\big(L^2(\bbR^2)\big), \quad 
z, z' \in \bbC \backslash [0,\infty),    \lb{2.77}
\end{split} 
\end{align} 
since also $|\bsB^{\prime}|^{1/2} (\bsH_0 - z \, \bsI)^{-1} \in \cB_2\big(L^2(\bbR^2)\big)$,  
$z \in \bbC \backslash [0,\infty)$, in analogy to \eqref{2.74}. 
\end{proof}

\section{The Computation of  $\xi(\, \cdot \, ; \bsH_2, \bsH_1)$} \lb{s5}

Given the results of Sections \ref{s3} and \ref{s4}, and Appendix \ref{sA}, we now turn to 
the computation of $\xi(\, \cdot \, ; \bsH_2, \bsH_1)$ and the proof of the key result \eqref{6}.

\begin{theorem}\lb{t4.4}
Assume Hypothesis \ref{h3.1}. Then, for $($Lebesgue\,$)$ a.e.~$\lambda > 0$ and 
a.e.~$\nu \in \bbR$, $($cf.\ \eqref{B.21}$)$,
\begin{equation}
\xi(\lambda; \bsH_2, \bsH_1) = \xi (\nu; A_+, A_-) = \f{1}{2 \pi} \int_{\bbR} dx \, \phi(x).    \lb{4.1}
\end{equation}
\end{theorem}
\begin{proof} 
Due to the fact that \eqref{2.57}, \cite{CGPST14} and \cite{Pu08} apply,  we have the approximate trace formula,
\begin{align}
\begin{split} 
& \tr_{L^2(\bbR^2)}\big((\bsH_{2,n} - z \, \bsI)^{-1}-(\bsH_{1,n} - z \, 
\bsI)^{-1}\big)    \\
& \quad = \f{1}{2z} \tr_{L^2(\bbR)} \big(g_z(A_{+,n})-g_z(A_-)\big),   \quad 
n \in \bbN, \; z\in \bbC \backslash [0,\infty),     \lb{trn}
\end{split} 
\end{align} 
with
\begin{equation}   
g_z(x) = x(x^2-z)^{-1/2}, \quad z\in\C\backslash [0,\infty), \; x\in\bbR.   
\end{equation}

Relation \eqref{trn} and the Krein--Lifshits trace formula yield  
\begin{equation}
\int_{[0,\infty)} \f{\xi(\lambda; \bsH_{2,n}, \bsH_{1,n}) d\lambda}{(\lambda - z)^2} 
= \f{1}{2} \int_{\bbR} \f{\xi(\nu; A_{+,n}, A_-) d\nu}{(\nu^2 -z)^{3/2}}, \quad 
n \in \bbN, \; z \in \bbC \backslash [0,\infty).      \lb{krn}
\end{equation}
As shown in the course of the proof of  Theorem\ 8.2 in \cite{GLMST11}, \eqref{krn} 
implies the relation 
\begin{align} 
& \int_{[0,\infty)} \xi(\lambda; \bsH_{2,n}, \bsH_{1,n}) d \lambda \, 
\big[(\lambda - z)^{-1} - (\lambda - z_0)^{-1}\big]     \no \\
& \quad = \int_{\bbR} \xi(\nu; A_{+,n}, A_-) d\nu \, \big[(\nu^2 - z)^{-1/2} - (\nu^2 - z_0)^{-1/2}\big],    \lb{krtrn} \\
& \hspace*{5cm}  n \in \bbN, \; z, z_0 \in \bbC \backslash [0,\infty).   \no
\end{align}

Combining Theorem \ref{t3.7} with the Krein--Lifshits trace formula 
\eqref{3.45} (for the pair $(\bsH_2,\bsH_1)$ as well as the pairs  
$(\bsH_{2,n}, \bsH_{1,n})$, $n \in \bbN$), yields
\begin{align}
& \lim_{n\to\infty} \int_{[0,\infty)} \f{\xi(\lambda; \bsH_{2,n}, \bsH_{1,n}) 
d\lambda}{(\lambda - z)^2}    \no \\
& \quad = - \lim_{n\to\infty} \tr_{L^2(\bbR^2)}\big((\bsH_{2,n} - z \, \bsI)^{-1} 
- (\bsH_{1,n} - z \, \bsI)^{-1}\big)    \no \\
& \quad = - \tr_{L^2(\bbR^2)}\big((\bsH_2 - z \, \bsI)^{-1} - (\bsH_1 - z \, \bsI)^{-1}\big)  \no \\
& \quad = \int_{[0,\infty)} \f{\xi(\lambda; \bsH_2, \bsH_1) 
d\lambda}{(\lambda - z)^2}, \quad z \in \bbC \backslash \bbR.       \lb{trlimH} 
\end{align}
Lemma \ref{l3.5}\,$(i)$ and Theorem \ref{t3.8} imply that the pairs of self-adjoint operators 
$(\bsH_{2,n}, \bsH_{1,n})$, $n \in \bbN$, and $(\bsH_2,\bsH_1)$ satisfy the hypotheses 
\eqref{detcont}, \eqref{det},\eqref{B2conv}, \eqref{B1conve} (identifying the pairs 
$(A_n, A_{0,n})$ and $(A,A_0)$ with the pairs $(\bsH_{2,n}, \bsH_{1,n})$ and 
$(\bsH_2,\bsH_1)$, respectively). Thus, an application of \eqref{detlim} to the pairs 
$(\bsH_{2,n}, \bsH_{1,n})$ and $(\bsH_2,\bsH_1)$ implies 
\begin{align}
\begin{split} 
& \lim_{n\to\infty} \int_{[0,\infty)} \xi(\lambda; \bsH_{2,n}, \bsH_{1,n}) 
d\lambda \, \big[(\lambda - z)^{-1} - (\lambda - z_0)^{-1}\big]     \\
& \quad = \int_{[0,\infty)} \xi(\lambda; \bsH_2, \bsH_1) 
d\lambda \, \big[(\lambda - z)^{-1} - (\lambda - z_0)^{-1}\big] , 
\quad z, z_0 \in \bbC \backslash \bbR.     \lb{trH} 
\end{split} 
\end{align}
On the other hand, since 
\begin{equation}
\big[(\nu^2 - z)^{-1/2} - (\nu^2 - z_0)^{-1/2}\big] \underset{|\nu| \to \infty}{=} 
\Oh\big(|\nu|^{-3}\big), \quad z, z_0 \in \bbC \backslash \bbR, 
\end{equation}
and hence $(\nu^2 + 1) \big[(\nu^2 - z)^{-1/2} - (\nu^2 - z_0)^{-1/2}\big]$ is uniformly 
bounded for $\nu \in \bbR$, \eqref{B.119} yields 
\begin{align} 
\begin{split}
& \lim_{n \to \infty} \int_{\bbR} \xi(\nu; A_{+,n}, A_-) d\nu \, 
\big[(\nu^2 - z)^{-1/2} - (\nu^2 - z_0)^{-1/2}\big]    \\
& \quad = \int_{\bbR} \xi(\nu; A_+, A_-) d\nu \, \big[(\nu^2 - z)^{-1/2} - (\nu^2 - z_0)^{-1/2}\big], 
\quad z, z_0 \in \bbC \backslash [0,\infty).    \lb{xiA} 
\end{split}
\end{align}
Thus, combining \eqref{krtrn}, \eqref{trH}, and \eqref{xiA} one finally obtains  
\begin{align} 
& \int_{[0,\infty)} \xi(\lambda; \bsH_2, \bsH_1) d \lambda 
\, \big[(\lambda - z)^{-1} - (\lambda - z_0)^{-1}\big]    \no \\
& \quad = \lim_{n \to \infty} \int_{[0,\infty)} \xi(\lambda; \bsH_{2,n}, \bsH_{1,n}) d \lambda 
\, \big[(\lambda - z)^{-1} - (\lambda - z_0)^{-1}\big]    \no \\
& \quad = \lim_{n \to \infty} \int_{\bbR} \xi(\nu; A_{+,n}, A_-) d\nu \, 
\big[(\nu^2 - z)^{-1/2} - (\nu^2 - z_0)^{-1/2}\big]    \no \\
& \quad = \int_{\bbR} \xi(\nu; A_+, A_-) d\nu \, \big[(\nu^2 - z)^{-1/2} - (\nu^2 - z_0)^{-1/2}\big], 
\quad z, z_0 \in \bbC \backslash [0,\infty).    \lb{xixi} 
\end{align}
At this point we invoke that
$\xi(\nu; A_+, A_-) = \f{1}{2 \pi} \int_{\bbR} dx \, \phi(x) := c_0$, $\nu \in \bbR$, and hence \eqref{xixi} reduces to 
\begin{align} 
\begin{split}
& \int_{[0,\infty)} \xi(\lambda; \bsH_2, \bsH_1) d \lambda 
\, \big[(\lambda - z)^{-1} - (\lambda - z_0)^{-1}\big]     \\
& \quad = c_0 \int_{\bbR} d\nu \big[(\nu^2 - z)^{-1/2} - (\nu^2 - z_0)^{-1/2}\big], 
\quad z, z_0 \in \bbC \backslash [0,\infty).    \lb{xiH} 
\end{split} 
\end{align}
Employing 
\begin{equation}
\int^x dy \, (y^2 - z)^{-1/2} = \ln\big(2(x^2 - z)^{1/2} + 2x\big) + C, \quad z \in \bbC \backslash [0,\infty)
\end{equation}
(see, e.g., \cite[No.\ 2.261]{GR80}), \eqref{xiH} implies  
\begin{equation} 
\int_{[0,\infty)} \xi(\lambda; \bsH_2, \bsH_1) d \lambda \, 
\big[(\lambda - z)^{-1} - (\lambda - z_0)^{-1}\big] = - c_0 \, \ln(z/z_0), 
\quad z, z_0 \in\bbC \backslash [0,\infty).     \lb{xiHc0}
\end{equation}
The elementary fact, 
\begin{equation}
c_0 \int_{[0,\infty)} d \lambda 
\big[(\lambda - z)^{-1} - (\lambda - z_0)^{-1}\big] = - c_0 \, \ln(z/z_0), 
\quad z, z_0 \in\bbC \backslash [0,\infty),   
\end{equation}
together with the uniqueness of the measure in functions with a representation such as 
\eqref{xiHc0} (e.g., via the Stieltjes inversion formula, see the discussion in \cite{AD56} and 
in \cite[Appendix~B]{We80}) applied to the a.c. measures  
$\xi(\lambda; \bsH_2, \bsH_1) \, d \lambda$ and $c_0 \, d\lambda$, respectively,  
then yields $\xi(\lambda; \bsH_2, \bsH_1) = c_0$ for a.e.~$\lambda > 0$, 
completing the proof of \eqref{4.1}.  
\end{proof}

\begin{remark} \lb{r4.5}
Although we did not have to use this in the proof of Theorem \ref{t4.4}, we note 
that \eqref{krn} also implies the approximate version of a Pushnitski-type formula, 
\begin{equation}
\xi(\lambda; \bsH_{2,n}, \bsH_{1,n}) = \f{1}{\pi} \int_{- \lambda^{1/2}}^{\lambda^{1/2}} 
\f{\xi(\nu; A_{+,n}, A_-) d \nu}{(\lambda - \nu^2)^{1/2}} \, \text{ for a.e.~$\lambda > 0$, 
$n\in\bbN$.}
\end{equation}
${}$ \hfill $\Diamond$
\end{remark}

\section{The Witten Index} \lb{s6}

In this section we briefly discuss the Witten index for the $(1+1)$-dimensional model 
under consideration following the detailed treatment in \cite{CGPST14a}.

\begin{definition} \lb{d8.1} 
Let $T$ be a closed, linear, densely defined operator in $\cH$ and  
suppose that for some $($and hence for all\,$)$ 
$z \in \bbC \backslash [0,\infty)$,  
\begin{equation} 
\big[(T^* T - z I_{\cH})^{-1} - (TT^* - z I_{\cH})^{-1}\big] \in \cB_1(\cH).   \lb{8.1} 
\end{equation}  
Then introducing the resolvent regularization 
\begin{equation}
\Delta_r(T, \lambda) = (- \lambda) \tr_{\cH}\big((T^* T - \lambda I_{\cH})^{-1}
- (T T^* - \lambda I_{\cH})^{-1}\big), \quad \lambda < 0,        \lb{8.2} 
\end{equation} 
the resolvent regularized Witten index $W_r (T)$ of $T$ is defined by  
\begin{equation} 
W_r(T) = \lim_{\lambda \uparrow 0} \Delta_r(T, \lambda),      \lb{8.3}
\end{equation}
whenever this limit exists. 
\end{definition} 

Here, in obvious notation, the subscript ``$r$'' indicates the use of the resolvent 
regularization (for a semigroup or heat kernel regularization we refer to \cite{CGPST14a}).
Before proceeding to compute the Witten index for the $(1+1)$-dimensional model, 
we recall the known consistency between the Fredholm and Witten index 
whenever $T$ is Fredholm:

\begin{theorem} $($\cite{BGGSS87}, \cite{GS88}.$)$  \lb{t8.2} 
Suppose that $T$ is a Fredholm operator in $\cH$. If \eqref{8.1} holds, then the 
resolvent regularized Witten index $W_r(T)$ exists, equals the Fredholm index, 
$\ind (T)$, of $T$, and
\begin{equation} 
W_r(T) =  \ind (T) = \xi(0_+; T T^*, T^* T).    \lb{8.4}
\end{equation}
\end{theorem}

Since $\bsD_\bsA^{}$ is not a Fredholm operator in $L^2(\bbR^2)$, we now determine 
the resolvent regularized Witten index of $\bsD_\bsA^{}$ as follows:

\begin{theorem} 
Assume Hypothesis \ref{h3.1}. Then $W_r(\bsD_\bsA^{})$ exists and equals 
\begin{equation}
W_r(\bsD_\bsA^{}) = \xi(0_+; \bsH_2, \bsH_1) = 
\xi(0; A_+, A_-) = \f{1}{2 \pi} \int_{\bbR} dx \, \phi(x).     \lb{8.5}
\end{equation}
\end{theorem}
\begin{proof}
Even though Hypothesis~2.1 in \cite{CGPST14a} is not satisfied for the 
$(1+1)$-dimensional model at hand, the fact (established in Theorem \ref{t4.4}) 
that for a.e.~$\lambda > 0$ and a.e.~$\nu \in \bbR$, 
\begin{equation}
\xi(\lambda; \bsH_2, \bsH_1) = \xi(\nu; A_+, A_-) = \f{1}{2 \pi} \int_{\bbR} dx \, \phi(x),  
\lb{8.6}
\end{equation}
trivially also implies the identity 
\begin{equation} 
\xi(\lambda; \bsH_2, \bsH_1) = \frac{1}{\pi}\int_0^{\lambda^{1/2}}
\frac{[\xi(\nu; A_+,A_-) + \xi(-\nu; A_+,A_-)] d \nu}{(\lambda-\nu^2)^{1/2}}, 
\quad \lambda > 0.     \lb{8.7}
\end{equation} 
Given equality \eqref{8.7} one can now follow the proof of 
\cite[Theorem~4.3]{CGPST14a} and obtain \eqref{8.5}.
\end{proof}

\appendix
\section{Some Facts On Spectral Shift Functions, Trace Formulas, and Modified 
Fredholm Determinants} \lb{sA}
\renewcommand{\theequation}{A.\arabic{equation}}
\renewcommand{\thetheorem}{A.\arabic{theorem}}
\setcounter{theorem}{0} \setcounter{equation}{0}

We recall a few basic facts on spectral shift functions employed in the 
bulk of this paper and provide results on trace formulas in terms of modified Fredholm 
determinants.

Closely following \cite[Sects.~2--6]{BY93} and \cite[Ch.~8]{Ya92}, we provide a brief 
discussion of how to restrict the open constant in the definition of the spectral shift function 
$\xi(\, \cdot \, ; A, A_0)$ up to an integer for a pair of self-adjoint operators $(A,A_0)$ in $\cH$ 
satisfying for some (and hence for all) $z_0 \in \rho(A) \cap \rho(A_0)$, 
\begin{equation}
\big[(A - z_0 I_{\cH})^{-1} - (A_0 - z_0 I_{\cH})^{-1}\big] \in \cB_1(\cH).    \lb{B.26a} 
\end{equation}
Motivated by the unitary Cayley transforms of $A$ and $A_0$, one introduces the modified 
perturbation determinant,
\begin{align}
\begin{split} 
\wti D_{A/A_0}(z;z_0) = {\det}_{\cH} 
\big((A - z I_{\cH})(A - \ol{z_0} I_{\cH})^{-1} (A_0 - \ol{z_0} I_{\cH})(A_0 - z I_{\cH})^{-1}\big),&  \\     
z \in \rho(A) \cap \rho(A_0), \; \Im(z_0) > 0,&     \lb{B.27a}
\end{split} 
\end{align}
and notes that (cf.\ \cite[p.270]{Ya92}) 
\begin{equation}
\ol{\wti D_{A/A_0} (z; z_0)} = \wti D_{A/A_0} (\ol z; z_0)/ \wti D_{A/A_0} (z_0; z_0), \quad 
 \wti D_{A/A_0} (\ol{z_0}; z_0) =1, 
\end{equation}
and 
\begin{align}
\begin{split} 
{\tr}_{\cH} \big[(A - z I_{\cH})^{-1} - (A_0 - z I_{\cH})^{-1}\big] 
= - \f{d}{dz} \ln\big( \wti D_{A/A_0} (z; z_0)\big),&  \\ 
z \in \rho(A) \cap \rho(A_0), \; \Im(z_0) > 0.& 
\end{split}
\end{align}
In addition,
\begin{equation}
\f{\wti D_{A/A_0} (z; z_0)}{\wti D_{A/A_0} (\ol{z}; z_0)} 
= \f{\wti D_{A/A_0} (z; z_1)}{\wti D_{A/A_0} (\ol{z}; z_1)}, \quad   
z \in \rho(A) \cap \rho(A_0), \; \Im(z_0) > 0, \, \Im(z_1) > 0. 
\end{equation}
Then, defining
\begin{align}
\begin{split}
& \xi(\lambda; A,A_0; z_0) = (2\pi)^{-1} \lim_{\varepsilon \downarrow 0} 
\big[\Im\big(\ln\big(\wti D_{A/A_0} (\lambda + i \varepsilon; z_0)\big)\big)  \lb{B.31a} \\
& \hspace*{4.2cm} - \Im\big(\ln\big(\wti D_{A/A_0} (\lambda - i \varepsilon; z_0)\big)\big)\big]
\, \text{ for a.e.~$\lambda \in \bbR$,} 
\end{split}
\end{align}
one obtains for $z \in \rho(A) \cap \rho(A_0)$, $\Im(z_0) > 0$, $\Im(z_1) > 0$, 
\begin{align}
&  \xi(\, \cdot \,; A,A_0; z_0) \in L^1\big(\bbR; (\lambda^2 + 1)^{-1} d\lambda\big),    \\
& \ln\big(\wti D_{A/A_0} (z; z_0)\big) = \int_{\bbR} \xi(\lambda; A,A_0; z_0) d\lambda 
\big[(\lambda -z)^{-1} - (\lambda - \ol{z_0})^{-1}\big],    \\
& \xi(\lambda; A,A_0; z_0) = \xi(\lambda; A,A_0; z_1) + n(z_0,z_1) \, 
\text{ for some $n(z_0,z_1) \in \bbZ$,}     \lb{B.34a} \\
& {\tr}_{\cH} \big[(A - z)^{-1} - (A_0 - z I_{\cH})^{-1}\big] 
= - \int_{\bbR} \f{\xi(\lambda; A,A_0; z_0) d \lambda}{(\lambda - z)^2},    \lb{B.35A} \\
& [f(A) - f(A_0)] \in \cB_1(\cH), \quad f \in C_0^{\infty}(\bbR),    \\
& {\tr}_{\cH} (f(A) - f(A_0)) = 
\int_{\bbR} \xi(\lambda; A,A_0; z_0) d\lambda \, f'(\lambda), 
\quad f \in C_0^{\infty}(\bbR)  
\end{align}
(the final two assertions can be greatly improved). For the origin of the celebrated 
Krein--Lifshits trace formula \eqref{B.35A} we refer, in particular, to \cite{Kr53}--\cite{Li56}. 

Up to this point $\xi(\, \cdot \,; A,A_0; z_0)$ has been introduced via \eqref{B.31a} and 
hence by \eqref{B.34a}, it is determined only up to an additive integer. It is possible to remove 
this integer ambiguity in $\xi(\, \cdot \,; A,A_0; z_0)$ by adhering to a specific normalization 
as follows: One introduces 
\begin{equation}
U_0(z_0) = (A_0 - z_0 I_{\cH})(A_0 - \ol{z_0} I_{\cH})^{-1}, \quad 
U(z_0) = (A - z_0 I_{\cH})(A - \ol{z_0} I_{\cH})^{-1}, \quad z_0 \in \bbC_+, 
\end{equation}
and then determines a normalized spectral shift function, denoted by 
$\hatt \xi(\, \cdot \,; A,A_0; z_0)$, with the help of fixing the branch of 
$\ln\big(\wti D_{A/A_0} (z_0 ; z_0)\big)$ by the equation,
\begin{align} 
\begin{split} 
i \Im\big(\ln\big(\wti D_{A/A_0} (z_0 ; z_0)\big)\big) &= 
2 i \Im(z_0) \int_{\bbR} \f{\hatt \xi(\lambda; A,A_0; z_0) d \lambda}{|\lambda - z_0|^2}   \\
&= {\tr}_{\cH} \big(\ln\big(U(z_0)U_0(z_0)^{-1}\big)\big).     \lb{B.36a}
\end{split} 
\end{align} 
Here $\ln(W)$, with $W$ unitary in $\cH$, is defined via the spectral theorem,
\begin{align}
\begin{split} 
W = \ointctrclockwise_{S^1} \mu \, dE_W(\mu),  \quad 
\ln(W) = i \arg(W) = i \ointctrclockwise_{S^1} \arg(\mu) \, dE_W(\mu),&   \\
\arg(\mu) \in (- \pi, \pi],& 
\end{split} 
\end{align}
with $E_W(\cdot)$ the spectral family for $W$, and 
\begin{equation}
\ln\big(U(z_0) U_0(z_0)^{-1}\big) 
= \ln\big(I_{\cH} + [U(z_0) - U_0(z_0)] U_0(z_0)^{-1}\big) \in \cB_1(\cH),
\end{equation}
since 
\begin{equation}
U(z_0) - U_0(z_0) = - 2i \Im(z_0) \big[(A - \ol{z_0} I_{\cH})^{-1} 
- (A_0 - \ol{z_0} I_{\cH})^{-1}\big] \in \cB_1(\cH). 
\end{equation}
In conjunction with \eqref{B.36a} we also mention the estimate,
\begin{equation}
\int_{\bbR} \f{|\hatt \xi(\lambda; A,A_0; z_0)\big| d \lambda}{|\lambda - z_0|^2} 
\leq \f{\pi}{2} \big\|(A - z_0 I_{\cH})^{-1} - (A_0 - z_0 I_{\cH})^{-1} \big\|_{\cB_1(\cH)}. 
\end{equation}

Moreover, if there exists $\alpha_0 \geq 0$ such that
\begin{equation}
\alpha \big\|(A - i \alpha I_{\cH})^{-1} - (A_0 - i \alpha I_{\cH})^{-1}\big\|_{\cB(\cH)} < 1, 
\quad \alpha > \alpha_0,     \lb{B.38a}
\end{equation}
then (cf.\ \cite[p.~300--303]{Ya92}), 
\begin{equation}
\hatt \xi(\, \cdot \,; A,A_0; i \alpha) = \hatt \xi (\, \cdot \, ; A,A_0) \, \text{ is independent of $\alpha$ 
 for $\alpha > \alpha_0$.}    \lb{B.39a}
\end{equation}

We also note that if 
\begin{equation}
\dom(A) = \dom(A_0), \quad (A-A_0)(A_0 - z_0 I_{\cH})^{-1} \in \cB_1(\cH)   \lb{B.40A}
\end{equation}
holds for some (and hence for all) $z_0 \in \rho(A_0)$, then \eqref{B.38a} is valid for 
$0 < \alpha_0$ sufficiently large, and (cf.\ \cite[p.~303--304]{Ya92}),  
\begin{equation}
\hatt \xi (\, \cdot \, ; A,A_0) \in L^1\big(\bbR; (|\lambda| + 1)^{-1 - \varepsilon} d \lambda\big), 
\quad \varepsilon > 0.    \lb{B.40a}
\end{equation}
Since different spectral shift functions only differ by a constant, the inclusion \eqref{B.40a} 
remains valid for all spectral shift functions under the assumptions \eqref{B.40A}. 

Finally, if 
\begin{equation} 
B=B^* \in \cB_1(\cH) \, \text{ and } \, A = A_0 + B,    \lb{B.41a} 
\end{equation} 
then \eqref{B.38a} holds with $\alpha_0 = 0$ and 
(cf.\ \cite[p.~303--304]{Ya92})
\begin{equation}
\hatt \xi (\, \cdot \, ; A,A_0) \in L^1\big(\bbR; d \lambda\big).    \lb{B.42a} 
\end{equation}
Assuming \eqref{B.41a}, one usually introduces the (standard) perturbation determinant,
 \begin{align}
 \begin{split} 
 D_{A/A_0} (z) &= {\det}_{\cH} \big((A - z I_{\cH}) (A_0 - z I_{\cH})^{-1}\big)    \\
 &= {\det}_{\cH} \big(I_{\cH} + B(A_0 - z I_{\cH})^{-1}\big), \quad 
 z \in \rho(A) \cap \rho(A_0),
 \end{split} 
 \end{align} 
and the associated spectral shift function
\begin{equation}
\xi(\lambda; A,A_0) = \pi^{-1} \lim_{\varepsilon \downarrow 0} 
\Im(\ln(D_{A/A_0}(\lambda + i \varepsilon))) \, \text{ for a.e.~$\lambda \in \bbR$,} 
\end{equation}
and hence obtains the following well-known facts for $z \in \rho(A) \cap \rho(A_0)$, 
$\Im(z_0) > 0$, 
\begin{align}
& \ol{D_{A/A_0} (z)} = D_{A/A_0} (\ol z),  \quad \lim_{|\Im(z)|\to\infty} D_{A/A_0} (z) =1,   \\
& \xi(\lambda; A,A_0) = (2\pi)^{-1} \lim_{\varepsilon \downarrow 0} \big[
\Im(\ln(D_{A/A_0}(\lambda + i \varepsilon))) 
- \Im(\ln(D_{A/A_0}(\lambda - i \varepsilon)))\big]   \no \\
& \hspace*{8.3cm} \text{ for a.e.~$\lambda \in \bbR$,}     \lb{B.45a} \\ 
& \wti D_{A/A_0}(z;z_0) = D_{A/A_0} (z) / D_{A/A_0} (\ol{z_0}),   \lb{B.46a} \\ 
& \xi(\, \cdot \, ; A,A_0) \in L^1(\bbR; d\lambda),    \\ 
& \ln(D_{A/A_0} (z)) = \int_{\bbR} \xi(\lambda; A, A_0) d \lambda \, (\lambda - z)^{-1},    \\
& \int_{\bbR} \xi(\lambda; A, A_0) d \lambda = {\tr}_{\cH} (B), \quad 
\int_{\bbR} |\xi(\lambda; A, A_0)| d \lambda \leq \|B\|_{\cB_1(\cH)},   
\end{align}
Combining the facts \eqref{B.31a}, \eqref{B.45a}, and \eqref{B.46a} at first only yields 
for some $n(z_0) \in \bbZ$, 
\begin{equation} 
\xi(\lambda; A,A_0; z_0) = \xi(\lambda; A,A_0) + n(z_0) \, \text{ for a.e.~$\lambda \in \bbR$.}
\end{equation}
However, also taking into account \eqref{B.39a} and \eqref{B.42a} finally yields 
\begin{equation} 
\hatt \xi(\lambda; A,A_0) = \xi(\lambda; A,A_0) \, \text{ for a.e.~$\lambda \in \bbR$.}  \lb{B.47a}
\end{equation}
Thus, the normalization employed in \eqref{B.36a} is consistent with the normalization implied by \eqref{B.42a}  in the case of trace class perturbations.

We continue this appendix with the following result, originally derived in \cite{GN12a} 
under slightly different hypotheses: 

\begin{theorem} \lb{tB.1}
$(i)$ Suppose $A_0$ and $A$ are self-adjoint operators with 
$\dom(A) = \dom(A_0) \subseteq \cH$, with $B = \ol{(A - A_0)} \in \cB(\cH)$. \\
$(ii)$ Assume that for some $($and hence for all\,$)$ $z_0 \in \rho(A_0)$,
\begin{align}
& |B|^{1/2} (A_0 - z_0 I_{\cH})^{-1} \in \cB_2(\cH),    \lb{B.2} 
\end{align}
and that 
\begin{equation}
\lim_{z \to \pm i \infty} 
\big\||B|^{1/2} (A_0 - z I_{\cH})^{-1} |B|^{1/2}\big\|_{\cB_2(\cH)} = 0. 
\lb{B.4}
\end{equation}
$(iii)$ Suppose that
\begin{equation}
{\tr}_{\cH} \big((A_0 - z I_{\cH})^{-1} B (A_0 - z I_{\cH})^{-1}\big) = 
\eta' (z), \quad z \in \rho(A_0),    \lb{B.5}
\end{equation} 
where $\eta(\cdot)$ has normal limits, 
$\lim_{\varepsilon \downarrow 0} \eta(\lambda + i \varepsilon) := \eta(\lambda + i0)$  
for a.e.\ $\lambda \in \bbR$.

Then
\begin{align}
& \int_{\bbR} \xi(\lambda; A, A_0) d\lambda 
\big[(\lambda - z)^{-1} - (\lambda - z_0)^{-1}\big] = \eta (z) - \eta (z_0)   \no \\
& \quad  
+ \ln\bigg(\f{{\det}_{2,\cH}\big(I_{\cH} + \sgn(B) |B|^{1/2} 
(A_0 - z I_{\cH})^{-1} |B|^{1/2}\big)}
{{\det}_{2,\cH}\big(I_{\cH} + \sgn(B) |B|^{1/2} (A_0 - z_0 I_{\cH})^{-1} |B|^{1/2}\big)}
\bigg), 
\quad z, z_0 \in \rho(A) \cap \rho(A_0),    \lb{B.6} 
\end{align} 
and for some constant $c\in\bbR$, 
\begin{align}
\begin{split} 
\xi(\lambda; A, A_0) &= \pi^{-1} \Im\big(\ln\big(
{\det}_{2,\cH}\big(I_{\cH} + \sgn(B) |B|^{1/2} 
(A_0 - (\lambda + i 0) I_{\cH})^{-1} |B|^{1/2}\big)\big)\big)   \lb{B.7} \\
& \quad + \pi^{-1} \Im(\eta(\lambda + i0)) + c \, \text{ for a.e.\ $\lambda \in \bbR$.} 
\end{split} 
\end{align}
\end{theorem} 

\begin{remark} \lb{rB.2}
$(i)$ We note that Theorem \ref{tB.1} was derived in \cite{GN12a} for unbounded 
quadratic form perturbations $B$ of $A_0$ and hence $A_0$ was assumed to be 
bounded from below. Since we here assume that $B$ is bounded, boundedness 
from below of $A_0$ is no longer needed and the proof of \cite[Theorem~2.3]{GN12a}
applies line by line to the current setting. We also note that since $B|_{\dom(A_0)}$ 
is symmetric and bounded, $B = \ol{(A - A_0)}$ is self-adjoint on $\cH$. The basic 
identity underlining Theorem \ref{tB.1} is, of course, 
\begin{align}
& {\tr}_{\cH}\big((A - z I_{\cH})^{-1} - (A_0 - z I_{\cH})^{-1}\big) + 
{\tr}_{\cH}\big((A_0 - z I_{\cH})^{-1}B(A_0 - z I_{\cH})^{-1}\big)    \lb{B.8} \\
& \quad = - \f{d}{dz} \ln\big({\det}_{2,\cH}\big(I_{\cH} + \sgn(B)|B|^{1/2}
(A_0 - z I_{\cH})^{-1}|B|^{1/2}\big)\big), \quad z \in \rho(A) \cap \rho(A_0). \no 
\end{align}
$(ii)$ One can show (cf.\ \eqref{det}) that 
\begin{equation}
\eta(z) - \eta(z_0) = (z - z_0) {\tr}_{\cH} \big((A_0 -z I_{\cH})^{-1} B (A_0 - z_0 I_{\cH})^{-1}\big), 
 \quad z, z_0 \in \rho(A) \cap \rho(A_0).
\end{equation}
We will use this fact later in this appendix. 
\hfill $\diamond$
\end{remark}

For modified Fredholm determinants and their properties, we refer, for instance, to 
\cite[Sect.~XI.9]{DS88}, \cite[Sect.\ IV.2]{GK69} and \cite[Ch.\ 9]{Si05}. Here we just note that 
\begin{equation}\lb{Z2.117}
{\det}_{2,\cH}(I_{\cK}-A) = \prod_{n\in\cJ}(1-\lambda_n(A))e^{\lambda_n(A)},\quad A\in \cB_2(\cH),
\end{equation}
where $\{\lambda_n(A)\}_{n\in \cJ}$ is an enumeration of the non-zero eigenvalues of $A$, listed in non-increasing order according to their moduli, and $\cJ\subseteq \bbN$ is an appropriate indexing set, and 
\begin{align}
& {\det}_{2,\cH}(I_{\cK}-A)= {\det}_{\cH}((I_{\cH}-A)\exp(A)), \quad A\in\cB_2 (\cH), \lb{Z3.21} \\
& {\det}_{2,\cH}((I_{\cK}-A)(I_{\cH}-B))={\det}_{2,\cH}(I_{\cK}-A){\det}_{2,\cH}(I_{\cH}-B)
e^{-\tr_{\cH}(AB)},\lb{Z3.22}\\
&\hspace*{8.65cm}A, B\in\cB_2(\cH).   \no 
\end{align} 
In addition, we recall the fact, that ${\det}_{2,\cH}(I_{\cH} + \cdot \,)$ 
is continuous on $\cB_2(\cH)$, explicitly, for some $c > 0$, one has the estimate (cf.\ \cite[Theorem~9.2\,(c)]{Si05})
\begin{align}
\begin{split} 
|{\det}_{2,\cH}(I_{\cH} + T) - {\det}_{2,\cH}(I_{\cH} + S)| \leq \|T - S\|_{\cB_2(\cH)} 
e^{c[\|S\|_{\cB_2(\cH)} + \|T\|_{\cB_2(\cH)} + 1]^2},&   \lb{detcont} \\
S, T \in \cB_2(\cH).&
\end{split} 
\end{align} 

In addition, we need some results concerning the connection between trace 
formulas and modified Fredholm determinants (cf., eg., \cite[Sect.\ IV.2]{GK69}, 
\cite[Ch.\ 9]{Si05}, \cite[Sect.\ 1.7]{Ya92}). Suppose that $A_0$ is self-adjoint in 
the complex, separable Hilbert space $\cH$, and assume that the self-adjoint operator 
$B$ in $\cH$, satisfies
\begin{align}
\begin{split} 
& \dom(B) \supseteq \dom(A_0),    \lb{domB} \\
& \text{$B$ is infinitesimally bounded with respect to $A_0$,}   
\end{split}
\end{align} 
and for some (and hence for all) $z_0, z_1 \in \rho(A_0)$, 
\begin{equation}
B (A_0 - z_0 I_{\cH})^{-1} \in \cB_2(\cH), \quad 
(A_0 - z_0 I_{\cH})^{-1} B (A_0 - z_1 I_{\cH})^{-1} \in \cB_1(\cH).      \lb{BB2}
\end{equation} 
(We recall that $B$, with $\dom(B) \supseteq \dom(A_0)$, is called infinitesimally bounded 
with respect to $A_0$, if for all $\varepsilon > 0$, there exists $\eta(\varepsilon) > 0$, such 
that for $f \in \dom(A_0)$, 
$\|B f\|_{\cH} \leq \varepsilon \|A_0 f\|_{\cH} + \eta(\varepsilon) \|f\|_{\cH}$.) 

Then by assumption \eqref{domB}, 
\begin{equation} 
A = A_0 + B, \quad \dom(A) = \dom(A_0),    
\end{equation}
is self-adjoint in $\cH$, and 
\begin{align}
\begin{split} 
& \big[(A - z I_{\cH})^{-1} - (A_0 - z I_{\cH})^{-1}\big] 
= \big[(A_0 - z I_{\cH})^{-1} B (A_0 - z I_{\cH})^{-1}\big]     \\
& \quad \times \big[(A_0 - z I_{\cH}) (A - z I_{\cH})^{-1}\big] \in \cB_1(\cH), \quad 
z \in \rho(A) \cap \rho(A_0).
\end{split} 
\end{align}
Given this setup, one concludes the trace formula (cf., e.g., \cite[p.~44]{Ya92})
\begin{align}
& \tr_{\cH} \big((A - z I_{\cH})^{-1} - (A_0 - z I_{\cH})^{-1}\big)      \no \\
& \quad = - \f{d}{dz} \ln \big({\det}_{2,\cH}\big(I_{\cH} 
+ B (A_0 - z_0 I_{\cH})^{-1}\big)\big)   \no \\
& \qquad - {\tr}_{\cH}\big((A_0 - z I_{\cH})^{-1} B (A_0 - z I_{\cH})^{-1}\big)   \lb{detAA0} \\
& \quad = - \int_{\bbR} \f{\xi(\lambda; A, A_0) d \lambda}{(\lambda - z)^2}, 
\quad z \in \rho(A) \cap \rho(A_0),    \no
\end{align}
and consequently, also
\begin{align}
& \ln \bigg(\f{{\det}_{2,\cH}\big(I_{\cH} + B (A_0 - z I_{\cH})^{-1}\big)} 
{{\det}_{2,\cH}\big(I_{\cH} + B (A_0 - z_0 I_{\cH})^{-1}\big)}\bigg)  
\no \\
& \qquad + (z - z_0) {\tr}_{\cH}\big((A_0 - z I_{\cH})^{-1} B (A_0 - z_0 I_{\cH})^{-1}\big)    
\lb{det} \\ 
& \quad = \int_{\bbR} \xi(\lambda; A, A_0) d \lambda \big[(\lambda - z)^{-1} 
- (\lambda - z_0)^{-1}\big], \quad z,z_0 \in \rho(A) \cap \rho(A_0).     \no 
\end{align}
(To verify \eqref{det} it suffices to differentiate either side of \eqref{det} w.r.t. $z$, comparing 
with the final three lines of relation \eqref{detAA0}, and observing that either side of \eqref{det} vanishes at $z = z_0$.)   

At this point we recall that ${\det}_{2,\cH}(I_{\cH} + \cdot \,)$ is continuous on $\cB_2(\cH)$, 
as recorded earlier in \eqref{detcont}. 

Next, suppose that $A_{0,n}$, $B_n$, $A_n = A_{0,n} + B_n$, $n \in \bbN$, and 
$A_0$, $B$, $A = A_0 + B$ satisfy hypotheses \eqref{domB} and \eqref{BB2}. Moreover, 
assume that for some (and hence for all) $z_0 \in \bbC \backslash \bbR$, 
\begin{align}
& \wlim_{n \to \infty} (A_{0,n} - z_0 I_{\cH})^{-1} = (A_0 - z_0 I_{\cH})^{-1},     \lb{wrl} \\
& \lim_{n\to\infty} \big\|B_n (A_{0,n} - z_0 I_{\cH})^{-1} -
B (A_0 - z_0 I_{\cH})^{-1}\big\|_{\cB_2(\cH)} = 0,     \lb{B2conv} \\ 
& \lim_{n\to\infty} \big\|(A_{0,n} - z_0 I_{\cH})^{-1} B_n (A_{0,n} - z_0 I_{\cH})^{-1}     \no \\
& \hspace*{1.1cm} - (A_0 - z_0 I_{\cH})^{-1}B (A_0 - z_0 I_{\cH})^{-1}\big\|_{\cB_1(\cH)} = 0.    \lb{B1conv} 
\end{align}
One notes that due to self-adjointness of $A_{0,n}, A_0$, $n \in \bbN$, relation \eqref{wrl} 
is actually equivalent to strong resolvent convergence, that is,
\begin{equation}
\slim_{n \to \infty} (A_{0,n} - z I_{\cH})^{-1} = (A_0 - z I_{\cH})^{-1},  
\quad z \in \bbC \backslash \bbR.    \lb{srl}
\end{equation}   
Moreover, the well-known identity (see, e.g., \cite[p.~178]{We80}), 
\begin{align}
& (T_1 - z I_{\cH})^{-1} - (T_2 - z I_{\cH})^{-1} = (T_1 - z_0 I_{\cH})(T_1 - z I_{\cH})^{-1}   \no \\
& \quad \times \big[(T_1 - z_0 I_{\cH})^{-1} - (T_2 - z_0 I_{\cH})^{-1}\big] 
(T_2 - z_0 I_{\cH})(T_2 - z I_{\cH})^{-1},     \lb{resid} \\
& \hspace*{6.6cm} z, z_0 \in \rho(T_1) \cap \rho(T_2),    \no 
\end{align}
where $T_j$, $j=1,2$, are linear operators in $\cH$ with 
$\rho(T_1) \cap \rho(T_2) \neq \emptyset$, together with Lemma \ref{l3.6}, \eqref{B1conv}, 
and \eqref{srl} imply
\begin{align}
& \lim_{n\to\infty} \big\|(A_{0,n} - z_0 I_{\cH})^{-1} B_n (A_{0,n} - z_1 I_{\cH})^{-1}     \no \\
& \hspace*{1.1cm} - (A_0 - z_0 I_{\cH})^{-1}B (A_0 - z_1 I_{\cH})^{-1}\big\|_{\cB_1(\cH)} = 0, 
\quad z_0, z_1 \in \bbC \backslash \bbR.    \lb{B1conve} 
\end{align}

Then \eqref{det} applied to the self-adjoint pairs $(A_n, A_{0,n})$, $n \in \bbN$, and $(A, A_0)$, in combination with \eqref{detcont}--\eqref{B1conve} implies the continuity result,  
\begin{align}
& \lim_{n\to\infty} \int_{\bbR} \xi(\lambda; A_n, A_{0,n}) d \lambda \, \big[(\lambda - z)^{-1} 
- (\lambda - z_0)^{-1}\big]     \no \\
& \quad =\lim_{n\to\infty}\Bigg\{ \ln \Bigg(\f{{\det}_{2,\cH}\big(I_{\cH} 
+ B_n (A_{0,n} - z I_{\cH})^{-1}\big)} 
{{\det}_{2,\cH}\big(I_{\cH} + B_n (A_{0,n} - z_0 I_{\cH})^{-1}\big)}\Bigg)   
\no \\
& \hspace*{1.8cm} + (z - z_0) 
{\tr}_{\cH}\big((A_{0,n} - z I_{\cH})^{-1} B_n (A_{0,n} - z_0 I_{\cH})^{-1}\big) 
\Bigg\}    \no \\
& \quad = \ln \Bigg(\f{{\det}_{2,\cH}\big(I_{\cH} 
+ B (A_0 - z I_{\cH})^{-1}\big)} 
{{\det}_{2,\cH}\big(I_{\cH} + B (A_0 - z_0 I_{\cH})^{-1}\big)}\Bigg)   
\no \\
&  \qquad + (z - z_0) 
{\tr}_{\cH}\big((A_0 - z I_{\cH})^{-1} B (A_0 - z_0 I_{\cH})^{-1}\big) 
\no \\
& \quad = \int_{\bbR} \xi(\lambda; A, A_0) d \lambda \, \big[(\lambda - z)^{-1} 
- (\lambda - z_0)^{-1}\big],  \quad z, z_0 \in \bbC \backslash \bbR.     \lb{detlim} 
\end{align}

We note that these considerations naturally extend to more complex situations where 
$A = A_0 +_q B$, $A_n = A_{0,n} +_q B_n$, $n \in \bbN$, are defined as quadratic form 
sums of $A_0$ and $B$ and $A_{0,n}$ and $B_n$ (without assuming any correlation 
between the domains of $A$, $A_n$ and $A_0$), and the modified Fredholm determinants 
are replaced by symmetrized ones as in Theorem \ref{tB.1}, see, for instance, 
\cite{GN12}, \cite{GN12a}, and \cite{GZ12}. Since we do not need this at this point, we 
omit further details. 

\medskip 
 
\noindent 
{\bf Acknowledgments.} We are indebted to Harald Grosse, Jens Kaad, 
Yuri Latushkin, Matthias Lesch  (in particular for sending us \cite{Wo07}), 
Konstantin Makarov, Alexander Sakhnovich, and Yuri Tomilov for helpful 
discussions and correspondence. We also thank the anonymous referee 
for a careful reading of our manuscript and for providing numerous 
helpful comments. 

A.C., F.G., D.P., and F.S. thank the Erwin Schr\"odinger International 
Institute for Mathematical Physics (ESI), Vienna, Austria, for funding support 
for this collaboration in form of a Research-in Teams project, ``Scattering Theory 
and Non-Commutative Analysis'' for the duration of June 22 to July 22, 2014. 

F.G. and G.L. are indebted to Gerald Teschl and the ESI for a kind invitation to visit the University of Vienna, Austria, for a period of two weeks in June/July of 2014. 

A.C., G.L., D.P., F.S., and D.Z.\ gratefully acknowledge financial support from the Australian Research Council. A.C.\ also thanks the Alexander von Humboldt 
Stiftung and colleagues at the University of M\"unster.

 
\end{document}